\documentclass{article}

\usepackage{microtype}
\usepackage{graphicx}
\usepackage{subfigure}
\usepackage{booktabs} 

\usepackage{hyperref}

\usepackage[Algorithm]{algorithm}
\usepackage[noend]{algorithmic}

\usepackage[accepted]{icml2024}

\usepackage{amsmath}
\usepackage{amssymb}
\usepackage{mathtools}
\usepackage{amsthm}

\usepackage[normalem]{ulem}

\usepackage[capitalize,noabbrev]{cleveref}

\theoremstyle{plain}
\newtheorem{theorem}{Theorem}[section]

\newtheorem{corollary}[theorem]{Corollary}
\theoremstyle{definition}

\theoremstyle{remark}
\newtheorem{remark}[theorem]{Remark}

\def\R{{\mathbb{R}}}
\def\F{{\mathbb{F}}}

\def\E{{\mathbb{E}}}

\def\F{{\mathcal{F}}}

\def\D{{\mathcal{D}}}

\def\H{{\mathcal{H}}}
\def\L{{\mathcal{L}}}

\def\cls{{{\rm{cls}}}}
\def\fc{{\rm{fc}}}
\def\new{{{\rm{new}}}}
\def\dist{{{\rm{dis}}}}

\def\sep{{\,|\,}}

\def\text#1{\rm{#1}}

\DeclareMathOperator*{\argmin}{arg\,min}

\usepackage[textsize=tiny]{todonotes}
\usepackage{url}

\icmltitlerunning{Toward Availability Attacks in 3D Point Clouds}

\begin{document}

\twocolumn[
\icmltitle{Toward Availability Attacks in 3D Point Clouds}



\icmlsetsymbol{equal}{*}

\begin{icmlauthorlist}
\icmlauthor{Yifan Zhu}{equal,amss,ucas}
\icmlauthor{Yibo Miao}{equal,amss,ucas}
\icmlauthor{Yinpeng Dong}{tsinghua,realai}
\icmlauthor{Xiao-Shan Gao}{amss,ucas,kaiyuan}
\end{icmlauthorlist}

\icmlaffiliation{amss}{Academy of Mathematics and Systems Science,
Chinese Academy of Sciences
}

\icmlaffiliation{ucas}{University of Chinese Academy of Sciences
}

\icmlaffiliation{tsinghua}{Tsinghua University
}

\icmlaffiliation{realai}{RealAI
}

\icmlaffiliation{kaiyuan}{Kaiyuan International Mathematical Sciences Institute}

\icmlcorrespondingauthor{Xiao-Shan Gao}{xgao@mmrc.iss.ac.cn}
\icmlcorrespondingauthor{Yibo Miao}{miaoyibo@amss.ac.cn}

\icmlkeywords{Machine Learning, ICML}

\vskip 0.3in
]

\printAffiliationsAndNotice{\icmlEqualContribution}



\begin{abstract}
Despite the great progress of 3D vision, data privacy and security issues in 3D deep learning are not explored systematically.
In the domain of 2D images, many availability attacks have been proposed to prevent data from being illicitly learned by unauthorized deep models.
However, unlike images represented on a fixed dimensional grid, point clouds are characterized as unordered and unstructured sets, posing a significant challenge in designing an effective availability attack for 3D deep learning. 
%
In this paper, we theoretically show that extending 2D availability attacks directly to 3D point clouds under distance regularization is susceptible to the degeneracy, rendering the generated poisons weaker or even ineffective.  This is because in bi-level optimization, introducing regularization term can result in update directions out of control. 
To address this issue, we propose a novel Feature Collision Error-Minimization (FC-EM) method, which creates additional shortcuts in the feature space, inducing different update directions to prevent the degeneracy of bi-level optimization.
Moreover, we provide a theoretical analysis that demonstrates the effectiveness of the FC-EM attack. 
Extensive experiments on typical point cloud datasets, 3D intracranial aneurysm medical dataset, and 3D face dataset verify the superiority and practicality of our approach.
Code is available at 
https://github.com/hala64/fc-em.
\end{abstract}

\section{Introduction}


\begin{figure*}[t]
    \centering   
    \includegraphics[width=0.94\textwidth]
    {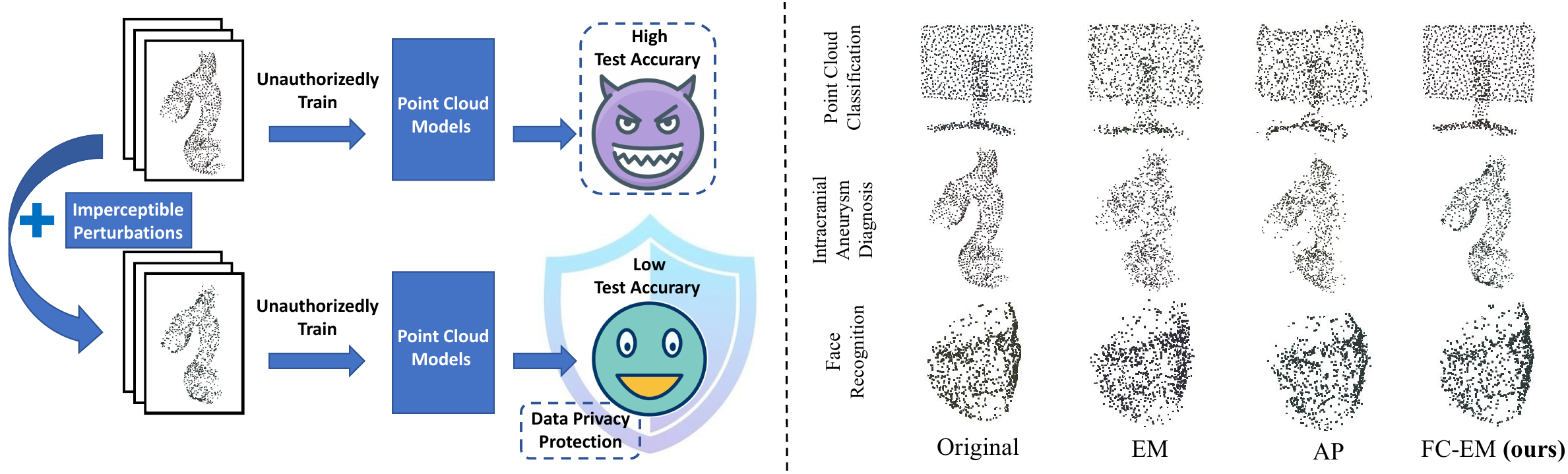}
    \vspace{-2ex}
    \caption{\textbf{Left:} An illustration of availability attacks. The poisoner adds imperceptible perturbations to the training data, aiming to reduce the model’s generalization ability and prevent data from being illicitly learned by unauthorized deep models. \textbf{Right:}
    An illustration of poisons crafted by EM \cite{huang2020unlearnable}, AP \cite{fowl2021adversarial}, and our FC-EM attack for point cloud classification, intracranial aneurysm diagnosis, and face recognition tasks.
    Notably, poisons generated by EM and AP have noticeable outliers, lacking imperceptibility.
    Poisons generated by our FC-EM are more natural and imperceptible, maintaining the semantic integrity.}    
    \label{fig:fig1}
    \vspace{-2ex}
\end{figure*}

Vision for 3D has been developed rapidly and become more popular in real-world applications, such as autonomous driving~\cite{chen2017multi,yue2018lidar}, medical image processing~\cite{taha2015metrics}, scene reconstruction~\cite{malihi20163d}, and face recognition~\cite{zhou20183d}. Since PointNet~\cite{qi2017pointnet} was first proposed, many deep learning-based methods have been applied to the 3D domain and shown tremendous success in various tasks.
Although great progress has been made, data privacy and security issues in 3D deep learning are not explored systematically.
Users may be reluctant to contribute their privacy-sensitive data, 
such as online healthcare records, to training large-scale commercial 3D models \cite{xu2023pointllm,zhou2023uni3d}.
In fact, according to a report by \citet{das2023healthdata},
an insurance company illicitly acquired a substantial amount of biomedical data to develop a commercial chronic disease risk prediction model. 
Another investigation indicated a growing number of lawsuits between data owners and machine learning companies~\cite{vincent2019google,burt2020facial,conklin2020facebook}.
Consequently, there is an increasing focus on safeguarding data from unauthorized use for training.

In the realm of 2D images, 
many availability attacks have been proposed to prevent data from being illicitly learned by unauthorized deep models~\cite{feng2019learning, huang2020unlearnable, fowl2021adversarial, Fu2022RobustUE, sandoval2022autoregressive}. 
They add imperceptible perturbations to the training data so that the model cannot learn much information from the data and thereby significantly reducing the model's accuracy on test data.
Within the literature, two representative 2D availability attack methods stand out. \citet{huang2020unlearnable} propose an error-minimization poisoning approach, known as Unlearnable Examples, leveraging iterative optimization of a bi-level min-min problem to create shortcuts by minimizing network loss. 
On the other hand, \citet{fowl2021adversarial} propose an error-maximization poisoning approach, known as Adversarial Poisons, 
employing stronger adversarial attacks on the pre-trained source classifier to generate poisoning noises.

However, to date, there is a notable absence of research in the domain of data privacy and security concerning 3D data. The data structure of 3D point clouds is inherently different from 2D images, posing a significant challenge in designing an effective availability attack.
Specifically, unlike images represented on a fixed dimensional grid, point clouds are characterized as unordered and unstructured sets.
As shown in Figure \ref{fig:fig1}, the poisons have noticeable outliers, lacking imperceptibility, 
thus making the $\ell_p$ norm constraints used in 2D availability attacks less suitable \cite{miao2022isometric}.
On the other hand, the Chamfer distance \cite{fan2017point} as an alternative metric is often used as a regularization term as it is difficult to directly constrain point clouds through projection onto a hypersphere due to the presence of the minimum operator. 
However, as theoretically analyzed in Section \ref{sec-reg-em-degrad}, extending 2D poison methods directly to 3D point clouds under distance regularization can be susceptible to the degeneracy of poisons, making the generated poisons weaker or even ineffective. 
This is because the complexity introduced by the regularization term in bi-level optimization leads to the optimization directions out of control.

To address this issue, we propose a novel method, \textbf{Feature Collision Error-Minimization (FC-EM)}. 
For the bi-level min-min problem of availability attacks, during the inner minimization, FC-EM employs our proposed feature collision loss rather than the original cross-entropy loss to optimize the poisoning noise, which creates additional shortcuts in the feature space to encourage the concentration of intra-class features and the dispersion of inter-class features.
By doing so, different optimization directions can prevent the degeneracy of bi-level optimization, thereby breaking the equilibrium. 
We also provide a theoretical analysis that demonstrates the effectiveness of the FC-EM attack. 
In summary, we propose the first availability attack specifically designed for point cloud classifiers 
and hope it could serve as a baseline for further studies into the data privacy and security issues of 3D deep learning.

Extensive experiments on typical point cloud recognition models \cite{qi2017pointnet,qi2017pointnet++,wang2019dynamic} demonstrate that, in comparison with other baseline poisoning methods,
our FC-EM method significantly reduces the model's generalization ability, while making the generated poisons more natural and imperceptible. 
Experiments on real-world datasets, including a 3D intracranial aneurysm medical dataset \cite{yang2020intra} and a 3D face dataset \cite{gerig2018morphable}, further confirms the effectiveness and practicality of our approach in real-world scenarios.

\vspace{-0.5ex}
\section{Related Work}
In this section, we introduce availability attacks and 3D point cloud attacks. 
More discussions on related work can be found in Appendix \ref{app-rel}.

\textbf{Availability attacks.}
%
Privacy issues have been extensively studied in the field of privacy-preserving machine learning \cite{shokri2015privacy, abadi2016deep, shokri2017membership}, including studies on availability attacks \cite{huang2020unlearnable, yu2022availability, sandoval2022poisons}. 
Availability attacks are a type of data poisoning, which allow the attacker to perturb the training dataset under a small norm restriction.
These attacks aim to cause test-time errors while maintaining the semantic integrity and not affecting the normal usage by legitimate users.
\citet{huang2020unlearnable} propose a bi-level error-minimizing approach to generate effective poisons, which is called unlearnable examples. 
\citet{fowl2021adversarial} use stronger adversarial poisons to achieve availability attacks.
Other availability attacks \cite{feng2019learning,yuan2021neural,yu2022availability,sandoval2022autoregressive,chen2023self, he2023sharpness} have also been proposed.
However, the problem of data privacy and security of 3D deep learning is still unexplored.
The data structure of 3D point clouds is inherently different from 2D images, posing a significant challenge in designing an effective availability attack method for 3D neural networks.

\textbf{Attacks on 3D point cloud.}
Vulnerability of 3D point cloud classifiers has become a grave concern due to the popularity of 3D sensors in applications.
Many works~\cite{xiang2019generating,cao2019adversarial,miao2022isometric} apply adversarial attacks to the 3D point cloud domain.
\citet{xiang2019generating} propose point generation attacks by adding a limited number of synthetic points. 
Further research~\cite{huang2022shape,liu2022imperceptible,tsai2020robust,zheng2019pointcloud,miao2022isometric} has employed gradient-based point perturbation attacks.
However, adversarial attacks 
are test-time evasion attacks. In contrast, our objective is to attack the model at train-time.
The backdoor attack is another type of attack that poisons training data with a stealthy trigger pattern~\cite{chen2017targeted}.
Some works \cite{li2021pointba,xiang2021backdoor,gao2023imperceptible} inject backdoors into 3D point clouds from a geometric perspective.
%
However, the backdoor attack does not harm the model’s performance on clean data.
Thus, it is not a valid method for data protection.
Different from these works, we introduce imperceptible perturbations into the training data, thereby reducing the model's generalization ability and achieving privacy protection.

\section{Preliminaries}

\subsection{Problem Formulation}

We formulate the problem of creating a clean-label poison in the context of 3D point cloud recognition. The poisoner has full access to the clean training dataset $D = \{(x_i, y_i)\}_{i=1}^N$, where $x_i = \{x_{i}^j\}_{j=1}^{n} \in \mathbb{R}^{n \times 3}$ represents the input point cloud and $y_i$ is the ground-truth label, and is able to add perturbations $\delta$ to each sample, releasing the poisoned version $D_{\delta} = \{(x_i+\delta_i, y_i)\}_{i=1}^N$.
The goal of the poisoner is to decrease the clean test accuracy of the point cloud model trained on $D_{\delta}$, while ensuring that the perturbation is imperceptible and can escape any detection from the victim. 
To this end, the poisoner constrains perturbations \( \delta \) by a certain budget \( \epsilon \), i.e., \( \lVert \delta \rVert_p \leq \epsilon \), where \(\lVert \cdot \rVert_p \) represents the $\ell_p$ norm. 
There are two representative availability attack methods.

\textbf{Error-minimization (EM)} \cite{huang2020unlearnable} involves iteratively optimizing a bi-level min-min problem to minimize the loss of the network, leading to the generation of poisoning noises, referred to as unlearnable examples:
\begin{align}
    \min_{{\theta}} \min_{\delta} \! \!\mathop{\E}\limits_{(x, y) \in D }\big[ \L_{\cls}(x+\delta,y;\theta)],  
    s.t. \,  \lVert \delta \rVert_p \leq \epsilon,  
\end{align}
where $\theta$ represents the model parameters, and $\L_{\cls}$ denotes the loss function. 
The outer minimization can imitate the training process, while the inner minimization can induce $\delta$ to have the property of minimizing the supervised loss.
Due to this property, deep models will pay more attention to the easy-to-learn $\delta$ and ignore $x$.

\textbf{Error-maximization (Adversarial Poisons, AP)} \cite{fowl2021adversarial} generates poisoning noises by conducting stronger adversarial attacks on a pre-trained source classifier, referred to as adversarial poisons:  
\begin{align} 
    &\max_{{\delta}} \mathop{\E}\limits_{(x, y) \in D }\big[ \L_{\cls}(x+\delta,y;\theta^*)\big], 
    s.t. \,  \lVert \delta \rVert_p \leq \epsilon , \nonumber\\
    & \ \ \ \  \theta^* \in \argmin_{\theta} \mathop{\E}\limits_{(x, y) \in D }\big[ \L_{\cls}(x, y;\theta)],
\end{align}
where $\theta^*$ represents the pre-trained model's parameters.
AP poisons have non-robust features in the incorrect class, causing DNNs generalize poorly on clean examples.

\subsection{3D Point Cloud Availability Attack}

However, in contrast to images which are fixed dimension grids, point clouds are represented as unordered and unstructured sets. As illustrated in Figure \ref{fig:fig1}, the poisons constrained by $\ell_p$ norms have noticeable point outliers, which are visually perceptible, making $\ell_p$ norm less suitable \cite{miao2022isometric}.
Instead, within 3D point clouds, \textit{Chamfer distance} \cite{fan2017point} often serves as an alternative to measure the distance of two point clouds, 
which is defined as:
$$\!\!\D_c(x, x')\! = \!\frac{1}{n} \! \sum\limits_{x_j\in x} \! \min\limits_{x_j'\in x'}\|x_j-x_j'\|_2^2 + \frac{1}{n'} \!\sum\limits_{x_j'\in x'} \! \min\limits_{x_j\in x}\|x_j'-x_j\|_2^2,$$
where $n$ and $n'$ are the number of points for $x$ and $x'$ respectively.
The Chamfer distance finds the nearest original point for each adversarial point and averages all the distances among the nearest point pairs.

Owing to the presence of the minimum operator in the Chamfer distance, it is difficult to directly constrain the point cloud through projection onto a hypersphere, akin to the $\ell_p$ norm. In the context of 3D point cloud adversarial attacks \cite{xiang2019generating,tsai2020robust,wen2020geometry}, Chamfer distance is commonly employed as a regularization term for perceptual distance. 
Thus, EM can be augmented by incorporating a distance regularization term, thus straightforwardly extending to:
\begin{align} \label{err-min-equ}
    \!\!\! \min_{{\theta}} \min_{\delta} \!\!\mathop{\E}\limits_{(x, y) \in D } \!\! \big[ \L_{\cls}(x+\delta,y;\theta) \!+\! \beta \cdot \L_{\dist} (x+\delta, x)\big],
\end{align}
where $\beta$ is a balancing hyperparameter between these two losses. We denote it as \textit{Regularized Error-Minimization} (REG-EM). 
Similarly, AP also can be extended:
\begin{align} \label{err-max-equ}
    &\max_{{\delta}} \mathop{\E}\limits_{(x, y) \in D }\big[ \L_{\cls}(x+\delta,y;\theta^*) - \beta \cdot \L_{\dist} (x+\delta, x)\big], \nonumber\\ 
    & s.t.\ \ \ \  \theta^* \in \argmin_{\theta} \mathop{\E}\limits_{(x, y) \in D }\big[ \L_{\cls}(x, y;\theta)].
\end{align}
We denote it as \textit{Regularized Adversarial Poisons} (REG-AP).
However, in the next section, we will show that Eq.~\eqref{err-min-equ} and Eq.~\eqref{err-max-equ} exhibit significant degeneracy issues, rendering the poisoning ineffective. 

\section{Poisoning Degeneracy}
\label{sec-reg-em-degrad}


While Eq.~\eqref{err-min-equ} and Eq.~\eqref{err-max-equ} straightforwardly extend the 2D poisoning methods to 3D point clouds under a distance regularization, this section demonstrates that such extension is susceptible to poisoning degeneracy, which makes the generated poisons weaker or even ineffective.

\textbf{Attack degeneracy on 3D unlearnable examples.}
%
%
In bi-level optimization, the occurrence of degeneracy phenomena, such as catastrophic forgetting and mode collapse observed in GAN \cite{thanh2020catastrophic}, is likely if the optimization directions are not well-controlled. Unfortunately, the bi-level optimization process REG-EM is not exempted from this degeneracy. Theorem~\ref{err-min-th} demonstrates that, 
the final convergence point of Eq.~\eqref{err-min-equ} compels the poison $\delta\to0$, rendering the poisoning ineffective.

%
%

\begin{theorem}
\label{err-min-th}
(Proof in Appendix \ref{proofs})
Assume that $\L_{\cls}$ and $\L_{\dist}$ are continuous, and the network's hypothesis space $\H_{\F}$ is compact. Let $D_{\delta}= \{(x_i+\delta_i, y_i)\}_{i=1}^N$ be the poisoned dataset of $D$. 
For simplicity, we denote $\mathop{\E}\limits_{(x, y) \in D }\big[ \L_{\cls}(x, y;\theta)]$ as $\L_{\cls}(D;\theta)$.
Then, there exists an optimal point ($\delta^*$, $\theta^*$) for the bi-level optimization \eqref{err-min-equ}, 
and the optimal perturbation $\delta^*$ satisfies $\L_{\dist}(D, \delta^*)\leq \frac{1}{\beta}\big[\min_{\theta} \L_{\cls}(D;\theta) -\min_{\delta, \theta} \L_{\cls}(D_{\delta};\theta)\big]$.
\end{theorem}

\vspace{3pt}
\begin{remark}
Theorem \ref{err-min-th} indicates that when the gap of $\min_{\theta} \L_{\cls}(D;\theta) -\min_{\delta, \theta} \L_{\cls}(D_{\delta};\theta)$ is extremely small, the optimal perturbation $\delta^*$ is subject to a very limited distance constraint, significantly reducing the poisoning power.
\end{remark}
\vspace{1pt}



\begin{corollary}
\label{err-min-cor}
(Proof in Appendix \ref{proofs})
Let $\L_{\cls}(x,y;\theta)=\max( \max_{t\neq y}f_{\theta}(x)_t - f_{\theta}(x)_y, 0)$, where $f_{\theta}(\cdot)$ is the output of logits layer, and $\L_{\dist}$ be a metric function. If  $\H_{\F}$ includes the function capable of correctly classifying $D$, then the equilibrium of the bi-level optimization \eqref{err-min-equ} will degenerate to $\delta^*=0$.
\end{corollary}

\vspace{-2pt}
%

Corollary \ref{err-min-cor} demonstrates that the optimal perturbation $\delta^*$ of Eq. \eqref{err-min-equ} degenerates to 0 when $\L_{\cls}$ is constructed as suggested by \citet{Carlini-Wagner2017}. 
In practice, when $\L_{\cls}$ is cross-entropy, $\min_{\theta} \L_{\cls}(D;\theta)$ is nearly zero. Consequently, the gap in Theorem \ref{err-min-th} is very small, compelling the degeneracy of poison $\delta$.
%
For REG-EM, i.e., Eq.~\eqref{err-min-equ}, we vary the regularization strength $\beta$ when generating poisons. The results in Figure \ref{reg-em-beta} suggest that, despite changes in $\beta$, there is no significant impact on the Chamfer distance and test accuracy. This implies that the poison has already degenerated when distance regularization is applied. 


 \begin{figure}[t]
 \vspace{-1ex}
 \hspace{-5mm}
 \centering
\subfigure[REG-EM]{\label{reg-em-beta}
\includegraphics[width=4.2truecm]{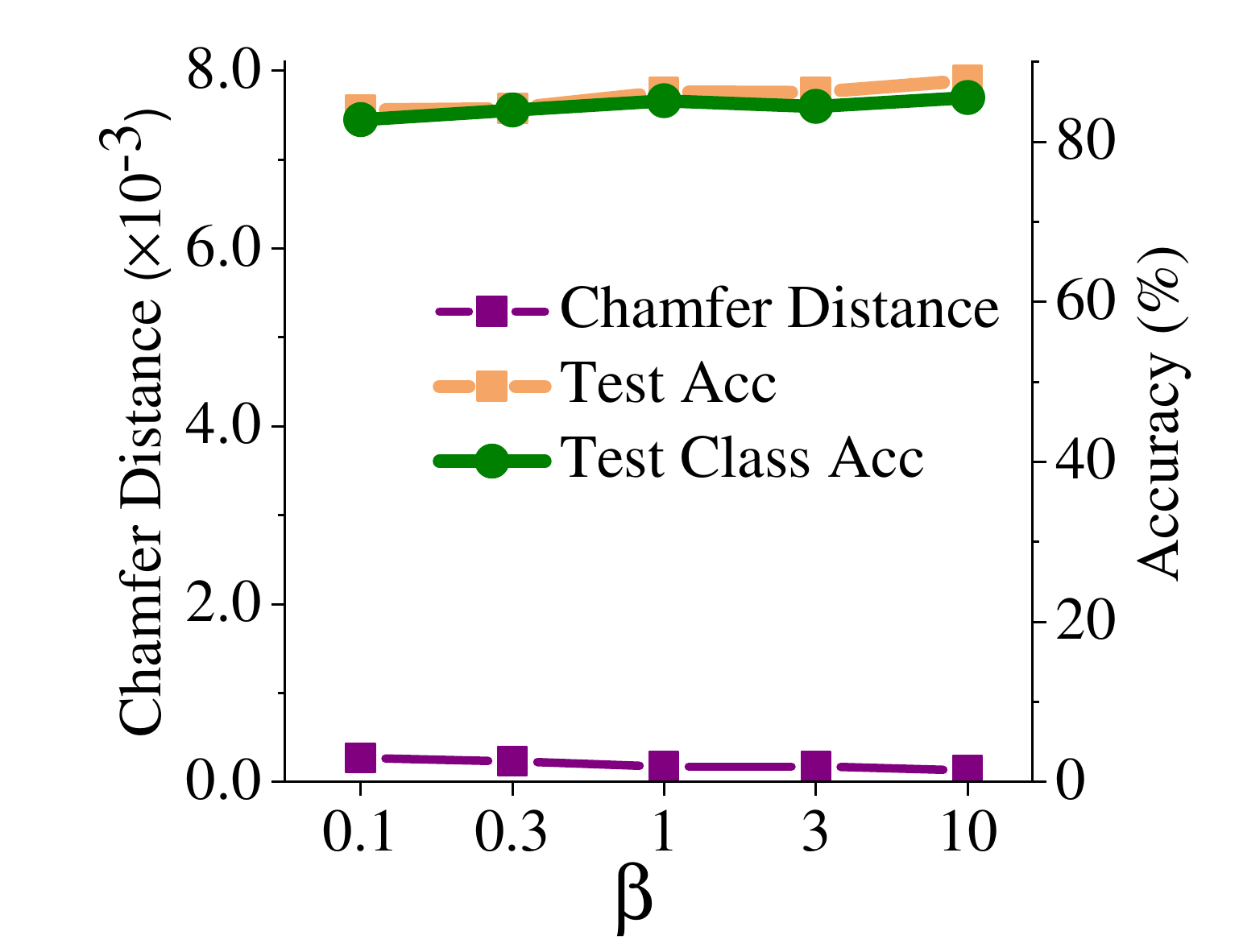}
}
 \hspace{-5mm}
\subfigure[FC-EM (\textbf{ours})]{\label{fc-em-beta}
\includegraphics[width=4.2truecm]{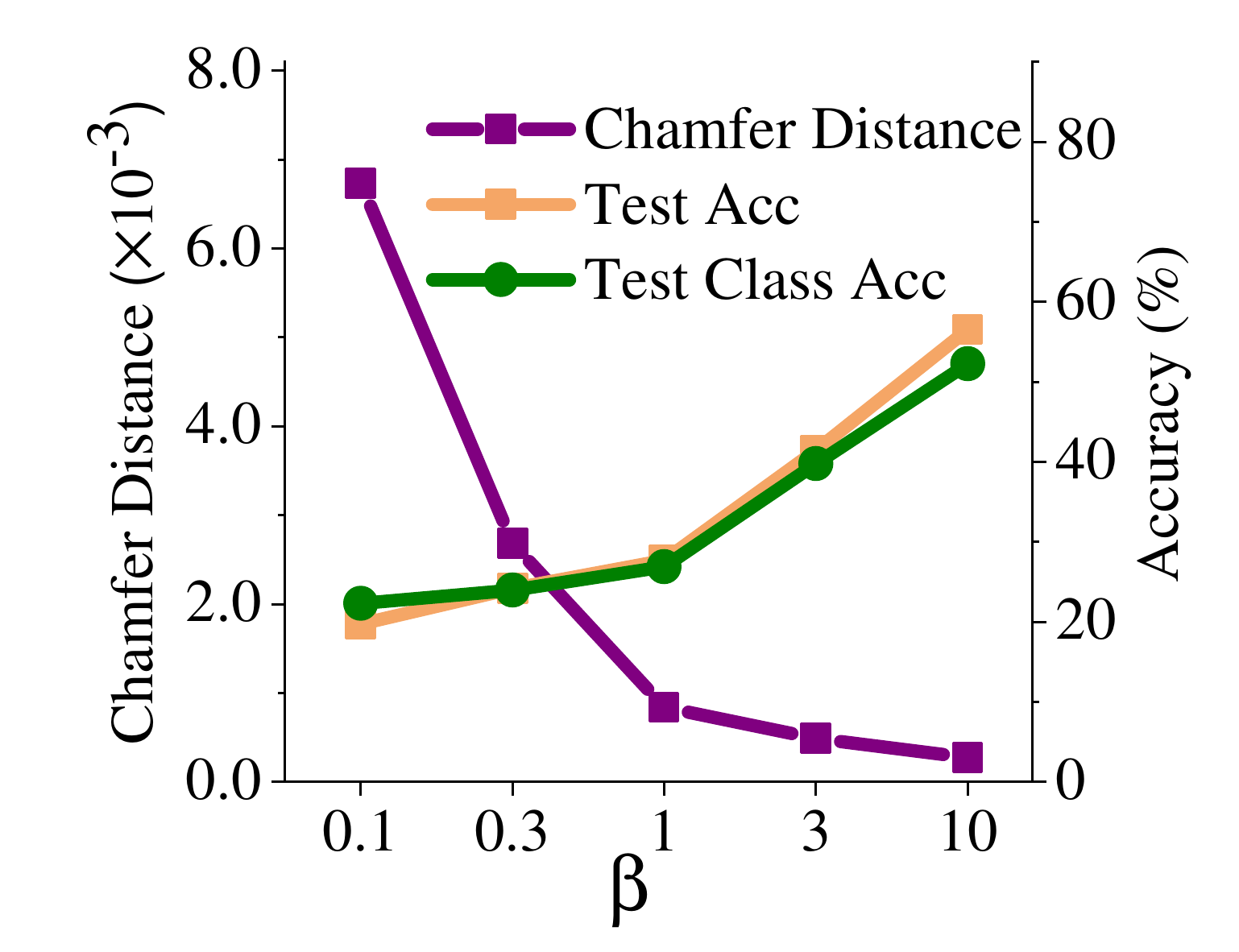}
}
 \vspace{-1mm}
\caption{The effects of different distance regularization strength $\beta$ under REG-EM and our FC-EM. 
Poisons crafted by REG-EM exhibit consistent Chamfer distance and test accuracy across different values of $\beta$.
This implies that the poison has already degenerated when distance regularization is applied.
In contrast, poisons crafted by FC-EM are more vulnerable to the balancing hyperparameter $\beta$, demonstrating resilience against poison degeneracy.
}
\vspace{-3mm}
\end{figure}


\textbf{Stronger attack becomes weaker for 3D adversarial poisons.}
AP aims to employ stronger adversarial examples for poison generation \cite{fowl2021adversarial}. 
However, when extending AP to REG-AP, the addition of a distance regularization term weakens $\delta$ rather than strengthens it. 
This occurs because, during the iterative process of REG-AP, once the adversarial attack succeeds, the optimization focus in Eq. \eqref{err-max-equ} shifts towards minimizing the regularization term, attempting to reduce $\delta$. 
This leads to weaker adversarial attacks compared to AP, which, under the constraint of $\ell_p$ norm, continues to seek stronger adversarial samples even after a successful attack.
Further detailed discussions will be deferred to Appendix \ref{ap-degrad}. 
%
This will significantly diminish the potency of the poisons, as demonstrated in the experiments in Section \ref{sec-exp}. Therefore, directly extending the AP process to 3D point clouds proves to be ineffective.

\section{Method}

As analyzed in Section \ref{sec-reg-em-degrad}, the inclusion of distance regularization in 3D point clouds would lead to degeneracy, significantly diminishing the efficacy.
Addressing the challenges arising from the complexity introduced by the distance regularization term in bi-level optimization becomes crucial.
In this section, we propose a novel method, \textbf{Feature Collision Error-Minimization (FC-EM)}, to mitigate this degeneracy and theoretically analyze its poisoning performance.

\subsection{Feature Collision Error-Minimization}
%
%

To generate potent poisons, a necessary condition is to disrupt the equilibrium provided in Theorem \ref{err-min-th} and Corollary~\ref{err-min-cor}.
To address this issue, our key insight is to employ a new loss for poison generation which is distinct from $\L_{\cls}$ in model optimization. By doing so, different update directions can prevent the degeneracy of bi-level optimization, 
thereby breaking the constraint provided in Theorem \ref{err-min-th}.

Specifically, due to the alternate iteration of bi-level optimization (\ref{err-min-equ}), we propose the use of a new loss $\L_{\textup{new}}$ during the optimization of $\delta$. For clarity, let $\L_1(\theta, \delta)=\E_{(x,y)\in D}\L_{\textup{cls}}(x+\delta, y; \theta) + \beta\cdot\L_{\dist}(x+\delta,x)$ and $\L_2(\theta, \delta)=\E_{(x,y)\in D}\L_{\textup{new}}(x+\delta, y; \theta) + \beta\cdot\L_{\dist}(x+\delta,x)$. $\theta$ is updated by minimizing $\L_1$, while $\delta$ is updated by minimizing $\L_2$. The subsequent theorem guides us in constructing a new loss to mitigate degeneracy. 
The results indicate that when there exists a substantial gap between $\L_2(\theta^*, 0)$ and $\L_2(\theta^*, \delta^*)$, the $\ell_2$ norm of $\delta^*$ is also significant, preventing degeneracy.

\vspace{1pt}

\begin{theorem}
\label{th-new-loss-2}
(Proof in Appendix \ref{proofs})
Let $D=\{(x_i, y_i)\}_{i=1}^N$, $\delta_i$ be the perturbations corresponding to $x_i$.
Denote the equilibrium for bi-level optimization $\min_{\theta,\delta}\{\L_1, \L_2\}$ as $(\theta^*, \delta^*)$.
Assume that $\L_{\textup{new}}$ is $L$-lipschitz with respect to $\delta$ under $\ell_2$ norm, $\L_{\dist}$ is Chamfer distance,  and for each point $x_i^j$, the $\ell_2$ norm of perturbation $\delta_i^j$ is bound by $\frac{1}{2}\min_{j,k}\lVert x_i^j-x_i^k \rVert_2$. 
Then it holds $\lVert \delta_i^* \rVert_2 \geq \frac{1}{4\beta}(L-\sqrt{L^2-8\beta\Delta})$, where $\Delta=\L_2(\theta^*, 0) - \L_2(\theta^*, \delta^*)$.
\end{theorem}



\vspace{1pt}

\begin{remark}
In light of Theorem \ref{th-new-loss-2},
when the gap $\Delta$ is substantial, the optimal perturbation $\delta_i$ under the $\ell_2$ norm also increases.
To maintain a substantial gap $\Delta$ and disrupt the equilibrium in bi-level optimization, the new loss $\L_{\textup{new}}$ should not converge to the optimal value $\theta^*$ under $\L_{\cls}$.
\end{remark}

\begin{algorithm}[t]
\caption{Feature Collision Error-Minimization Poisoning Attack (FC-EM)}
\label{alg:fc-em}\small
\begin{algorithmic}
\STATE{\bfseries Input:} A 3D point cloud training dataset $D=\{(x_i, y_i)\}_{i=1}^N$. 
Total epoch $T$. Batch size $N_B$. Distance loss $\L_{\dist}$ and regularization strength $\beta$.
Feature collision loss $\L_{\fc}$ and temperature $t$. 
Classifier parameters $\alpha_{\theta}$ and $T_\theta$. Attack parameters $\alpha_{\delta}$, $T_\delta$ and $T_a$.   

\STATE {\bfseries Output:} 
Poisoned dataset $D_{\delta}=\{ (x_i + \delta_i, y_i) \}_{i=1}^N$

\STATE {\bfseries Initialize:} $\delta_i\gets 0, i=1,2,\cdots,N$  

\FOR {$t=1,\cdots,T$}
    \FOR {$t_\theta = 1,\cdots, T_\theta$} 
        \STATE 
        Sample a mini batch $B=\{(x_{b_j}, y_{b_j})\}_{j=1}^{N_B}$. 
        \STATE 
        $\theta \gets \theta - \alpha_{\theta} \cdot \nabla_\theta  \E_{(x_{b_j}, y_{b_j})\in B } \big[ \L_{\cls}(x_{b_j} + \delta_{b_j}, y_{b_j};\theta) + \beta \cdot \L_{\dist} (x_{b_j}+\delta_{b_j}, x_{b_j})\big]$  
    \ENDFOR
    
    \FOR {$t_\theta = 1,\cdots, T_\delta$}  
    \STATE Sample a mini batch $B=\{(x_{b_j}, y_{b_j})\}_{j=1}^{N_B}$.
       \FOR{$t_a = 1,\cdots, T_a$}
       \STATE Compute class-wise feature collision loss $\L_{\fc}$.
        \STATE $\delta_{b_j} \gets \delta_{b_j} - \alpha_{\delta} \cdot \nabla_{\delta_{b_j}} \E_{(x_{b_j}, y_{b_j}) \in B } \big[ \L_{\fc}(x_{b_j}+\delta_{b_j},y_{b_j}; $\\ $ \theta, t) + \beta \cdot \L_{\dist} (x_{b_j}+\delta_{b_j}, x_{b_j})\big]$
        \ENDFOR
    \ENDFOR
\ENDFOR

\end{algorithmic}
\end{algorithm}

%

\vspace{-2pt}
Theorem \ref{th-new-loss-2} guides us to construct a new loss with distinct optimization directions.
Effective availability attacks often require creating shortcuts \cite{yu2022availability, sandoval2022poisons} that are so simplistic, being linearly separable, that deep models tend to rely on them for predictions, neglecting genuine features. 
We consider generating poisons containing shortcuts in the feature space rather than using $\L_{\cls}$ in the logit space to induce different update directions. 
Motivated by \citet{chen2020simple}, we propose a novel class-wise feature similarity loss $\L_{\fc}$ with its $j$-th dimension:\vspace{-1ex}
\begin{align} \label{fc-loss}
\hspace{1pt}
    & \!\!\!\L_{\fc}(x+\delta,y;\theta, t)_j = -\log \nonumber \\
    & \!\!\!\!\!\!\frac{\sum_{k\in B} \mathbb{I}( y_{b_k}\!=\!y_{b_j}\!) \exp(s(g(x_{b_j}\!+\!\delta_{b_j}\!), g(x_{b_k}\!+\!\delta_{b_k}\!))/t)}{\sum_{k\in B}\exp(s(g(x_{b_j}+\delta_{b_j}), g(x_{b_k}+\delta_{b_k}))/t)}, 
\end{align}
where $B=\{x_{b_j}, y_{b_j}\}_{j=1}^{N_B}$ is the mini-batch, $g$ represents the feature extractor of the network $f_\theta$, $s(\cdot, \cdot)= \frac{\left<\cdot, \cdot\right>}{\lVert \cdot\rVert_2\lVert\cdot\rVert_2}$ denotes the cosine similarity between two vectors, and $t$ is the temperature of the loss.
%
$\L_{\fc}$ represents inter-class feature similarity as a surrogate loss that reflects the clustering process, generating shortcuts for 3D point clouds.
Specifically, at each step $i$, we iteratively update $\delta$ and $\theta$:
\begin{align}
\indent\indent
    & \!\!\Delta_{\theta}^i \! = \! \argmin_{\theta} \! \mathop{\E}\limits_{ \!\!\!(x, y) \in D \!} \!
    \big[ \L_{\cls}(x \!+\! \delta^{i-\!1} \!,y;\theta)  \! + \! \beta \! \cdot \!  \L_{\dist} (x \!+ \! \delta^{i-\!1}\! ,\! x)\big], \nonumber \\ 
    & \!\Delta_{\delta}^i \! =  \!  \argmin_{{\delta}}  \mathop{\E}\limits_{ \!\!\!(x, y) \in D \!} \! \big[ \L_{\fc}(x \! + \! \delta,y;\theta^i, t) \! + \! \beta 
 \! \cdot \! \L_{\dist} (x \! + \! \delta, x)\big], \nonumber \\ 
    &\theta^{i} =\theta^{i-1} + \Delta_{\theta}^i, \delta^{i} = \delta^{i-1} + \Delta_{\delta}^i.
\end{align}\vspace{-5ex}
%
%
%

It is noteworthy that $\L_{\fc}$ modifies logits by replacing softmax and negative log-likelihood operations with cosine similarity. In the numerator of $\L_{\fc}$, we incorporate all samples from the same class, not just the logits corresponding to true labels. Therefore, the objective of minimizing Eq.~\eqref{fc-loss} is to increase feature similarity within the same class while decreasing feature similarity across different classes, which we referred to as {\em feature collision} by similarity.
We create shortcuts in the feature space, causing the network to focus more on the intra/inter-class features of $\delta$, ignoring those of $x$.
The detailed procedure of FC-EM is provided in Algorithm \ref{alg:fc-em}.

\begin{figure}[t]
\vspace{-2ex}
\hspace{-5mm}
    \centering
    \subfigure[]{
    \label{ce_fc_loss_curve}
\includegraphics[width=4.2truecm]{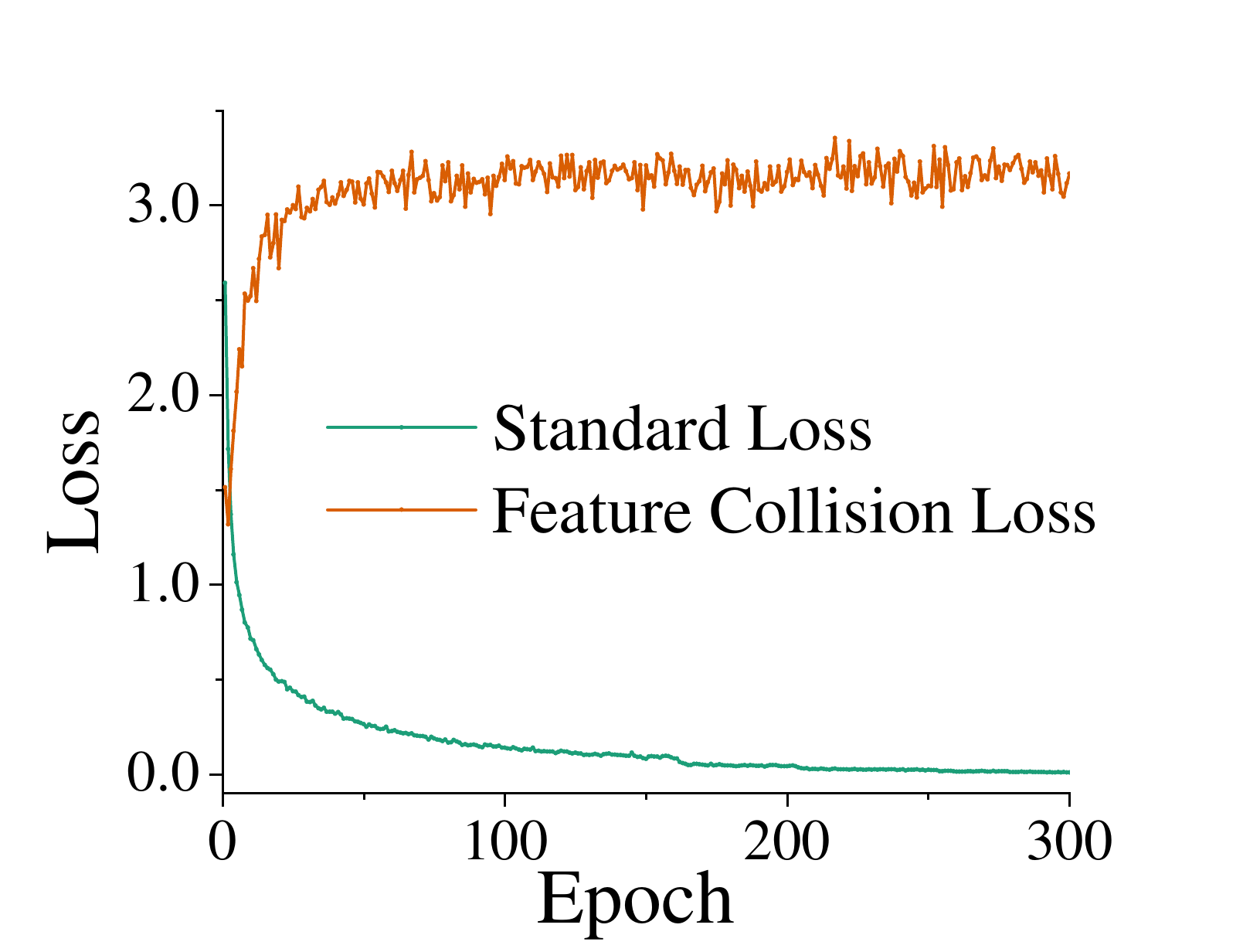}
    }
\hspace{-5mm}
    \subfigure[]{ \label{cossim_linear_layer}
\includegraphics[width=4.2truecm]{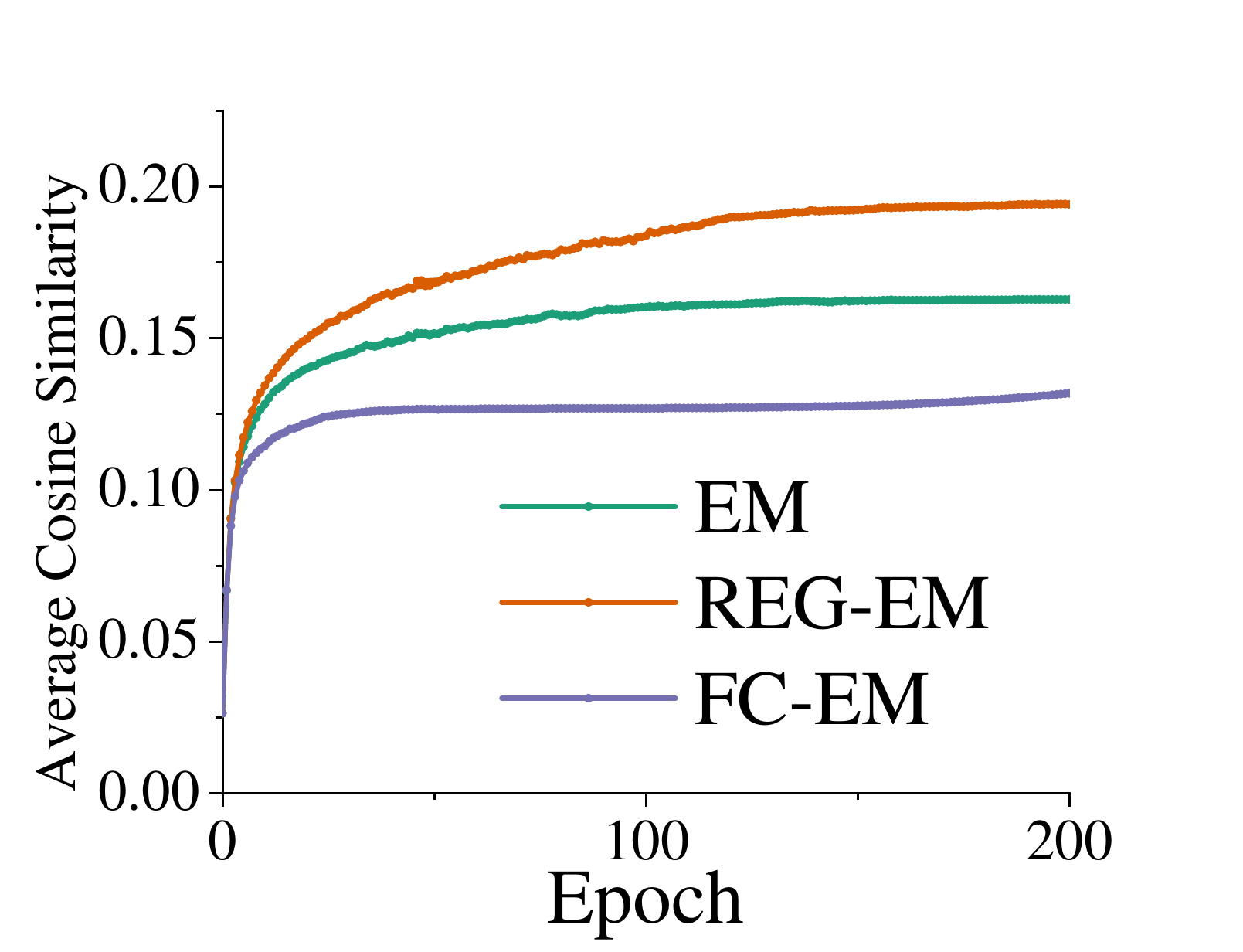}}
\vspace{-2mm}
    \caption{\textbf{(a):} Epoch-loss curves of cross-entropy loss and feature collision loss under standard training. The feature collision loss fails to converge and even increases. This implies that they optimize towards the different direction.
    \textbf{(b):} Average cosine similarity between the rows of last layer weight matrix. FC-EM yields smaller cosine similarities compared to EM and REG-EM.
    }
\vspace{-5mm}
\end{figure}

%
%

%

As shown in Figure \ref{ce_fc_loss_curve}, when vanilla training a network using the cross entropy loss  $\L_{\cls}$, the feature collision loss $\L_{\fc}$ fails to converge and even increases. 
In other words, while minimizing $\L_{\cls}$ and $\L_{\fc}$ both aim to create shortcuts, they do not optimize towards the same direction.
This discrepancy leads to a larger $\L_2(\theta^*, 0)$, as derived in Theorem~\ref{th-new-loss-2}, resulting in a larger $\Delta$ and subsequently causing a larger $\delta^*$ to break the degeneracy under distance regularization.
In comparison to REG-EM depicted in Figure \ref{reg-em-beta}, Figure \ref{fc-em-beta} reveals that FC-EM poisons are more vulnerable to the distance regularization strength $\beta$, demonstrating resilience against poison degeneracy. 





\begin{figure*}[t!]
\centering
\includegraphics[width=0.99\textwidth]
{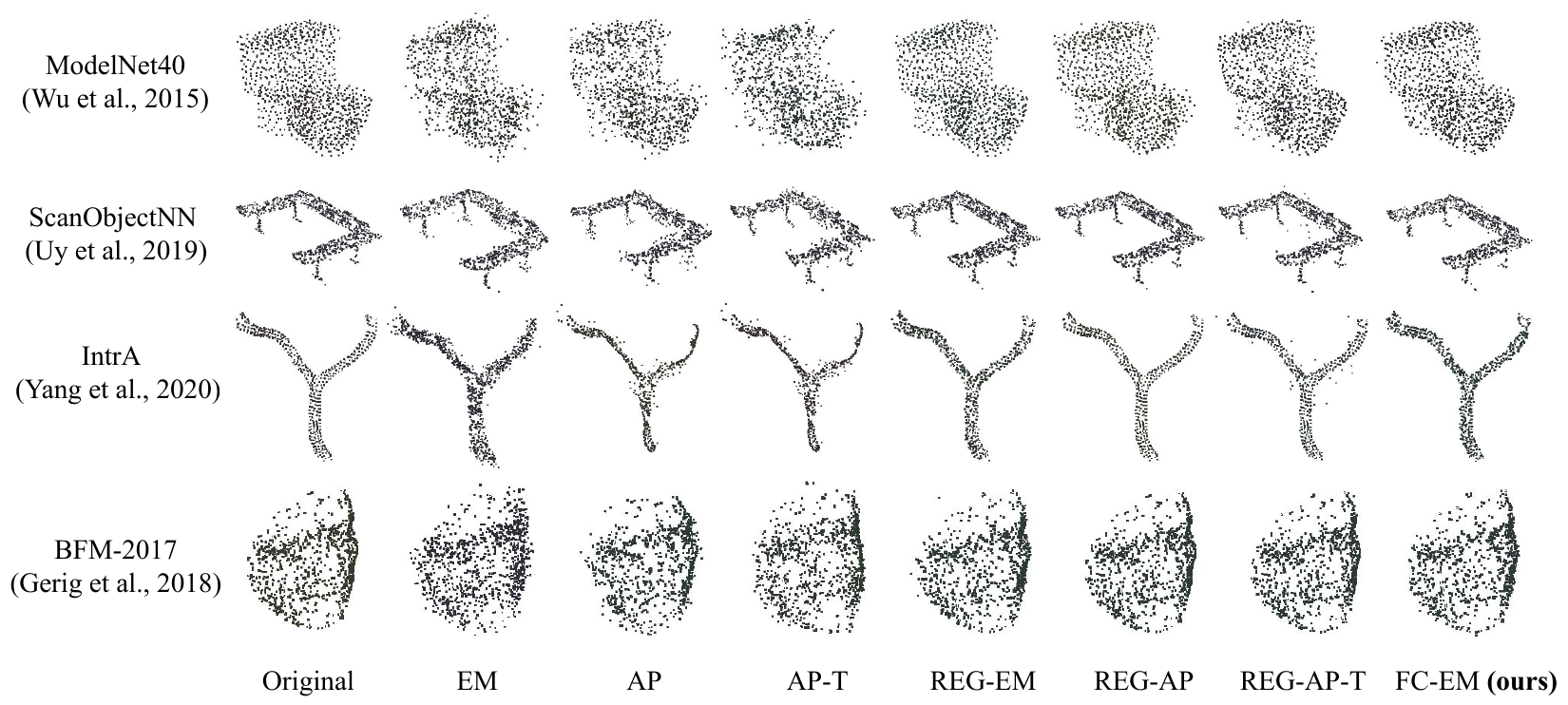}
\vspace{-2ex}
\caption{
Qualitative visualization results of baseline methods and our FC-EM.
Poisons generated by EM, AP, AP-T and REG-AP-T exhibit conspicuous outliers, thus lacking imperceptibility.
Although REG-EM and REG-AP successfully achieve imperceptibility, they fail to reduce model’s accuracy on test data.
In contrast, our FC-EM approach not only demonstrates enhanced naturalness and imperceptibility but also effectively reduce model's generalization ability.}
\label{fig:fig5_1_2}
\vspace{-3mm}
\end{figure*}
\subsection{Theoretical Analysis}
%
%
%

%

To assess the suitability of our proposed FC-EM attack for 3D point clouds, we provide a theoretical analysis that demonstrates the effectiveness of FC-EM attack.
\citet{yu2022availability, zhu2023detection} find both experimentally and theoretically that the effectiveness of availability attacks is closely tied to linear separability.
The reason why current availability attacks work may stem from that the imperceptible poisons create shortcuts. 
They are almost linearly separable, that deep models would overwhelmingly rely on spurious features for predictions.
Below, we demonstrate that FC-EM attack exhibits stronger linear separability compared to traditional EM and REG-EM attacks. Firstly, we establish a correlation between the linear separability of the poisoned dataset and the cosine similarity between the last layer weights.

\begin{theorem}
\label{th-linear}
(Proof in Appendix \ref{proofs})
Let 
 $\alpha$ be the loss minimum of dataset $D$ under all linear classifier, the weight matrix of a linear classifier be denoted as $W$, where $\ell_2$ norm of each row is normalized to $\sqrt{d}$. For any two rows of $W$, their cosine similarity does not exceed $\gamma$. The loss function is $\L_{\cls}(x,y;\theta)=\max( \max_{t\neq y}f_{\theta}(x)_t - f_{\theta}(x)_y, 0)$. 
Then, there exist poisons $\{\delta_i\}_{i=1}^N$ satisfying  $\mathop{\E}_{(x_i,y_i)\in D}\D_{\rm{c}}(x_i, x_i+\delta_i)\leq O(\frac{\alpha}{(1-\gamma)^2})$,
such that the poisoned dataset is linearly separable.
\end{theorem}






\begin{remark}
Theorem \ref{th-linear} implies that as the cosine similarity $\gamma$ decreases, the upper bound on the Chamfer distance of existing effective poisons becomes smaller.
 Consequently, to achieve linear separability in poisoned datasets, models with lower cosine similarity require less Chamfer distance as a poisoning budget, making poisons more effective.
\end{remark}

\vspace{-2pt}
FC-EM optimizes the feature collision loss by enhancing intra-class cosine similarity and reducing inter-class cosine similarity simultaneously.
The weight matrix $W$ of the final layer linear classifier has each row representing a distinct pattern match for a specified class, tending to align with the corresponding features.
Consequently, due to the consistency in cosine similarity between 
inter-class features and different rows in $W$
, FC-EM tends to diminish the cosine similarity of rows.
In fact, as illustrated in Figure \ref{cossim_linear_layer}, FC-EM indeed yields smaller cosine similarities compared to EM and REG-EM. According to Theorem \ref{th-linear}, this implies that if the Chamfer distance is constrained, FC-EM exhibits stronger linear separability than EM and REG-EM. Therefore, it possesses stronger poison, resulting in lower clean test accuracy of the victim model.

\section{Experiments}
\label{sec-exp}

\begin{table*}[t!]
\caption{
Quantitative results of baselines and our FC-EM on ModelNet40. 
Our FC-EM achieves the lowest test accuracy (Acc) and exhibits suitable imperceptibility (measured by Chamfer distance ($\D_c$) and Hausdorff distance ($\D_h$)), causing the strongest availability attack.
}
\small
\label{modelnet40-perf}
\centering
\setlength{\tabcolsep}{4pt}
\begin{tabular}{lccccccccc}
 \toprule
  & \multicolumn{3}{c}{PointNet} & \multicolumn{3}{c}{PointNet++} & \multicolumn{3}{c}{DGCNN}\\
  Poison  & Acc & $\D_c$($\times 10^{-4}$) & $\D_h$($\times 10^{-3}$) & Acc & $\D_c$($\times 10^{-4}$) & $\D_h$($\times 10^{-3}$) & Acc & $\D_c$($\times 10^{-4}$) & $\D_h$($\times 10^{-3}$)  \\
\midrule
  Clean & 90.70\% & - & - & 93.06\% & - & - & 92.64\% & - & - \\
    EM & 82.98\% & 29.7 & 14.1 & 82.41\% & 27.5 & 11.8 & 80.72\% & 27.2 & 11.6 \\
  AP & 69.37\% & 27.3 & 17.0 & 89.26\% & 25.6 & 11.1 & 78.04\% & 29.2 & 10.9\\
  AP-T & 61.95\% & 26.7 & 16.8 & 56.33\% & 22.8 & 9.7 & 57.45\% & 19.6 & 8.7\\
  REG-EM & 86.22\% & 1.7 & 7.8 & 87.52\% & 13.2 & 8.2 & 85.21\% & 6.5 & 5.4 \\
  REG-AP   & 75.32\% & 5.5 & 11.5 & 88.11\% & 5.4 & 3.4 & 88.49\% & 7.8 & 5.2 \\
  REG-AP-T  & 72.61\% & 5.1 & 11.5 & 77.88\% & 3.0 & 5.3 & 75.97\% & 3.3 & 5.1 \\
  FC-EM (\textbf{ours}) & \textbf{27.67\%} & 8.4 & 13.5 & \textbf{31.12\%} & 8.9 & 9.0 & \textbf{32.46\%} & 8.1 & 8.0\\
\bottomrule
\end{tabular}
\vspace{-2ex}
\end{table*}

\begin{table*}[t!]
\vspace{-1ex}
\caption{
Quantitative results of baselines and our FC-EM on ScanObjectNN. 
Our FC-EM obtains the lowest test accuracy and gains good imperceptibility (measured by Chamfer distance ($\D_c$) and Hausdorff distance ($\D_h$)), outperforming all other existing availability attacks.
}
\small
\label{scan-perf}
\centering
\setlength{\tabcolsep}{4pt}
\begin{tabular}{lccccccccc}
 \toprule
 & \multicolumn{3}{c}{PointNet} & \multicolumn{3}{c}{PointNet++} & \multicolumn{3}{c}{DGCNN}\\
  Poison  & Acc & $\D_c$($\times 10^{-4}$) & $\D_h$($\times 10^{-3}$) & Acc & $\D_c$($\times 10^{-4}$) & $\D_h$($\times 10^{-3}$) & Acc & $\D_c$($\times 10^{-4}$) & $\D_h$($\times 10^{-3}$)  \\
\midrule
  Clean & 74.87\% & - & - & 85.54\% & - & - & 81.07\% & - & - \\
  EM & 63.68\% & 32.5 & 12.7 & 66.78\% & 22.6 & 7.9 & 46.82\% & 22.3 & 8.0 \\
  AP & 49.91\% & 14.4 & 14.5 & 54.04\% & 11.2 & 7.0 & 32.01\% & 13.7 & 8.5 \\
  AP-T & 49.74\% & 13.5 & 13.9 & 34.77\% & 10.4 & 9.3 & 26.68\% & 13.1 & 8.2 \\
  
  REG-EM & 69.71\% & 2.3 & 1.8 & 72.46\% & 4.0 & 2.8 & 65.23\% & 4.0 & 2.8 \\
  REG-AP   & 67.13\% & 2.0 & 3.9 & 73.49\% & 4.7 & 3.9 & 56.63\% & 9.5 & 7.7 \\
  REG-AP-T  & 50.60\% & 3.2 & 8.1 & 52.32\% & 2.3 & 4.0 & 43.72\% & 5.9 & 7.9 \\
  FC-EM (\textbf{ours}) & \textbf{33.21\%} & 4.6 & 10.4 & \textbf{27.71\%} & 8.7 & 9.2 & \textbf{23.41\%} & 4.5 & 6.8 \\
\bottomrule
\end{tabular}
\vspace{-2ex}
\end{table*}


\subsection{Experimental Setup}
\textbf{Dataset.} 
In our experiments, we employ several datasets for evaluation, including the generated point clouds dataset, ModelNet40 \cite{wu20153d}, the scanned point clouds dataset, ScanObjectNN \cite{uy2019revisiting}, the 3D intracranial aneurysm medical dataset, IntrA \cite{yang2020intra}, and the BFM-2017 generated face dataset \cite{gerig2018morphable}. Details of these datasets are provided in Appendix \ref{exp-setting}. 

\textbf{Victim Models.} 
We select three commonly used point cloud classification networks as our victim models, including PointNet \cite{qi2017pointnet}, PointNet++ \cite{qi2017pointnet++} and DGCNN \cite{wang2019dynamic}. 

\textbf{Evaluation metrics.} 
To quantitatively assess the effectiveness of various attacks,
we train victim models on the poisoned training dataset, and evaluate the accuracy (Acc) on the clean test set. Besides, to measure the imperceptibility, 
we use the Chamfer distance ($\D_c$) and Hausdorff distance \cite{rockafellar2009variational} ($\D_h$) as evaluation metrics.

\textbf{Implementation details.}
We use PyTorch learning rate scheduler, 
ReduceLROnPlateau for both poisoning and evaluation process, with initial learning rate $10^{-3}$.
%
Except for above baseline methods, 
we also evaluate the targeted version of AP and REG-AP, namely AP-T and REG-AP-T.
The poisoning epoch for EM, REG-EM and FC-EM is set to 200, and for AP(-T), REG-AP(-T) is set to 100.
%
The evaluation epoch is set to 200.
The coefficient of regularization term $\beta$ is set to 1.0 originally.
%
Across all datasets, we uniformly sample 1,024 points for attacks and classification.
For further details, including more ablation studies and experiments, please refer to Appendices \ref{exp-setting} and \ref{app-exp}.

\vspace{-3pt}

\subsection{Main Results}

We compare the proposed FC-EM with with various baselines, including the distance regularization methods, REG-EM and REG-AP(-T), and the $\ell_{\infty}$-norm restraint methods, EM and AP(-T).
%
Table \ref{modelnet40-perf} shows the results on ModelNet40, 
our FC-EM 
achieves the lowest test accuracy (Acc) compared with all other baselines and exhibits suitable imperceptibility (measured by Chamfer distance $\D_c$) and Hausdorff distance $\D_h$).
Consequently, it exhibits the worst generalization and causes strongest poison.
Visualization on ModelNet40 under different methods is depicted in the first row of Figure \ref{fig:fig5_1_2}. 
Our FC-EM demonstrates great imperceptibility.
This is because FC-EM creates poisons under distance regularization rather than simple $\ell_p$ norm restriction, leading poisons to focus more on the structural aspects of point clouds.
%
In contrast, the $\ell_{\infty}$-norm restraint methods, EM and AP(-T) exhibit many outliers and irregular deformations, simultaneously inducing high $\D_c$ and $\D_h$. 
%
%
On the other hand, although REG-EM and REG-AP successfully achieve imperceptibility, they fail to reduce model’s accuracy on test data, remaining Acc even larger than 70\%.
%
In summary, our FC-EM achieves better availability attack on both performance and imperceptibility.

Evaluation results and corresponding visualizations on ScanObjectNN are provided in Table \ref{scan-perf} and the second row of Figure \ref{fig:fig5_1_2}, respectively. Results on ScanObjectNN reveal similar trends with those on ModelNet40.
FC-EM obtains lowest test accuracy and gains good imperceptibility. 
Thus it outperforms all other baselines.
More discussions are deferred to Appendix \ref{scan-app}.
%
%


\begin{table*}[t]
\caption{
Quantitative results of baselines and our FC-EM on 3D intracranial aneurysm medical dataset IntrA.
Our FC-EM attains the lowest F1-score while keeping imperceptibility (measured by $\D_c$ and $\D_h$), obtaining strongest availability attack.
}
\small
\label{intra-perf}
\centering
\setlength{\tabcolsep}{4pt}
\begin{tabular}{lccccccccc}
 \toprule
 & \multicolumn{3}{c}{PointNet} & \multicolumn{3}{c}{PointNet++} & \multicolumn{3}{c}{DGCNN}\\
  Poison  & F1-score & $\D_c$($\times 10^{-4}$) & $\D_h$($\times 10^{-3}$) & F1-score & $\D_c$($\times 10^{-4}$) & $\D_h$($\times 10^{-3}$) & F1-score & $\D_c$($\times 10^{-4}$) & $\D_h$($\times 10^{-3}$)  \\
\midrule
  Clean & 0.667 & - & - & 0.821 & - & - & 0.741 & - & - \\
  EM & 0.458 & 36.4 & 16.4 & 0.391 & 24.4 & 8.7 & 0.417 & 29.5 & 11.4 \\
  AP & 0.360 & 17.9 & 16.2 & 0.288 & 15.4 & 6.7 & 0.175 & 15.0 & 9.6\\
  AP-T & 0.324 & 17.9 & 16.1 & 0.288 & 15.4 & 6.7 & 0.179 & 15.1 & 9.7\\
  REG-EM & 0.471 & 2.1 & 2.1 & 0.545 & 4.1 & 3.3  & 0.485 & 4.0 & 2.9 \\
  REG-AP & 0.585 & 0.9 & 1.2 & 0.571 & 7.6 & 3.4 & 0.485 & 10.1 & 5.4 \\
  REG-AP-T & 0.410 & 1.6 & 4.5 & 0.747 & 0.4 & 1.5 & 0.207 & 0.7 & 2.9 \\
  FC-EM (\textbf{ours}) & \textbf{0.285} & 2.0 & 2.0 & \textbf{0.000} & 0.9 & 2.0 & \textbf{0.000} & 0.9 & 2.0 \\
\bottomrule
\end{tabular}
\vspace{-2ex}
\end{table*}

\subsection{Transferability between Models} 
We further verify the transferability of our FC-EM.
Results presented in Table \ref{trans-table} shows the decent transferability of FC-EM, as all the test accuracy remains consistently below 40\%.
Our FC-EM is stable, displaying small sensitivity in both the source and victim models.
This implies that FC-EM poison appears to manifest as intrinsic availability defects within the dataset itself rather than arising from shortcomings in specific models.

\begin{table}[t]
\vspace{-1ex}
\caption{
Quantitative results of model's transferability of our FC-EM on ModelNet40. 
Our FC-EM exhibits good transferability, with all the test accuracy remaining consistently below 40\%.
}
\label{trans-table}
\small
\centering
\begin{tabular}{lcccccc}
 \toprule
  Eva./Gen. Model   & PointNet & PointNet++ & DGCNN\\
\midrule
  PointNet  & 27.67\% & 38.01\% & 30.83\% \\
  PointNet++   & 35.05\% & 31.12\% & 33.95\% \\
  DGCNN   & 33.31\% & 37.88\% & 32.50\% \\
\bottomrule
\end{tabular}
\end{table}

\subsection{Performance under Defense}
We evaluate various availability attacks on a range of defense methods, including AT-based approaches, PGD \cite{madry2018towards} and TRADES \cite{zhang2019theoretically}, data augmentation, Mixup \cite{zhang2018mixup} and Rsmix \cite{lee2021regularization}, and a defense tailored for 3D point clouds, IF-Defense \cite{wu2020if}.
The results provided in Table \ref{tab-defense} demonstrates that FC-EM outperforms all other availability attacks in all defense methods.
 \begin{table}[t]
 \vspace{-1ex}
\caption{
Quantitative results of the test accuracy (Acc) under different defense methods.
Our FC-EM outperforms all other availability attacks in all defense methods.
}
\label{tab-defense}
\centering
\small
\setlength{\tabcolsep}{3pt}
\begin{tabular}{lcccccccc}
 \toprule
  Acc(\%)   & ST & PGD & TRADES & Mixup & Rsmix  
  & IF-D \\
\midrule
Clean & 90.70 & 82.62 & 78.77 & 89.95 & 89.26  
& 87.12 \\
  EM & 82.98 & 65.48 & 70.22 & 84.00 & 82.90 
  & 80.79\\
  AP & 69.37 & 78.97 & 76.94 & 76.05 & 71.27 
  & 75.81 \\
  AP-T & 61.95 & 75.97 & 72.97 & 67.02 & 63.70 
  & 76.05 \\
  REG-EM & 86.22 & 82.09 & 78.44 & 89.14 & 87.76 
  & 86.06 \\
  REG-AP & 75.32 & 74.35 & 77.39 & 82.58 & 80.88 
  & 84.93 \\
  REG-AP-T & 72.61 & 77.63 & 74.67 & 73.95 & 70.83 & 82.70\\
  FC-EM (\textbf{ours})   & \textbf{27.67} & \textbf{49.59} & \textbf{52.80} & \textbf{65.32} & \textbf{51.86}
  & \textbf{75.73} \\
\bottomrule
\end{tabular}
\end{table}

 

\subsection{Results on Real-World Datasets}

\textbf{Results on 3D medical dataset.}
We conduct evaluation on the medical dataset IntrA \cite{yang2020intra}, a 3D intracranial aneurysm dataset framed as binary classification problem for distinguishing aneurysms from healthy vessel segments. 
The evaluation is based on the F1-score.
The results provided in Table \ref{intra-perf} reveal that FC-EM attains the lowest F1-score while keeping distance small, and even reaches to zero F1-score when utilizing PointNet++ and DGCNN. In such scenarios, victim models' predictions collapse entirely, consistently yielding negative outputs regardless of the input, rendering the F1-score to be zero.
%
IntrA is visualized in the third row of Figure \ref{fig:fig5_1_2}. It illustrates that both regularization methods and FC-EM maintain the structural integrity of vessels. 
In contrast, $\ell_{\infty}$ methods largely disrupt the structure, introducing numerous outliers. 
Consequently, FC-EM not only preserves the imperceptibility, but also obtains strongest availability attack.

\textbf{Results on 3D face dataset.}
We employ Basel Face Model 2017 \cite{gerig2018morphable} as the source face models, generating 3D point clouds from facial scans. Details regarding the generation process are outlined in Appendix \ref{exp-setting}.
The poison results on PointNet and the visualizations are presented in Table \ref{face-stat} and the fourth row of Figure \ref{fig:fig5_1_2}, respectively.
Notably, FC-EM and regularization methods display superior face contour. In contrast, $\ell_{\infty}$ methods severely disrupt the contour and have more outliers. 
Moreover, compared to regularization methods, FC-EM induces stronger attack, prompting a substantial 40\% decrease in the test accuracy, succinctly showcasing its superiority. 
More ablation studies and experiments can be found in Appendix \ref{app-exp}.
 \begin{table}[t]
 \vspace{-1ex}
\caption{
Quantitative results of baselines and our FC-EM on the BFM-2017 generated 3D face dataset.
Our FC-EM attains the lowest test accuracy, inducing strongest availability attack.
}
\label{face-stat}
\small
\setlength{\tabcolsep}{5pt}
\centering
\begin{tabular}{lccccccccccc}
 \toprule
  Poison  & Acc & $\D_c$($\times 10^{-4}$) & $\D_h$($\times 10^{-3}$) \\
\midrule
  Clean  & 98.1\% & - & - \\
  EM  & 53.5\% & 23.8 & 11.4 \\
  AP  & 76.1\% & 10.1 & 11.9 \\
  AP-T & 46.9\% & 7.3 & 10.8 \\
  REG-EM & 50.9\% & 4.5 & 5.6  \\
  REG-AP & 79.3\% & 2.0 & 3.7 \\
  REG-AP-T &  91.1\% & 0.8 & 1.2 \\
  FC-EM (\textbf{ours}) & \textbf{11.0\%} & 8.5 & 11.3 \\
\bottomrule
\end{tabular}
\vspace{-1ex}
\end{table}


\begin{figure*}[t!]
\centering
\includegraphics[width=0.99\textwidth]
{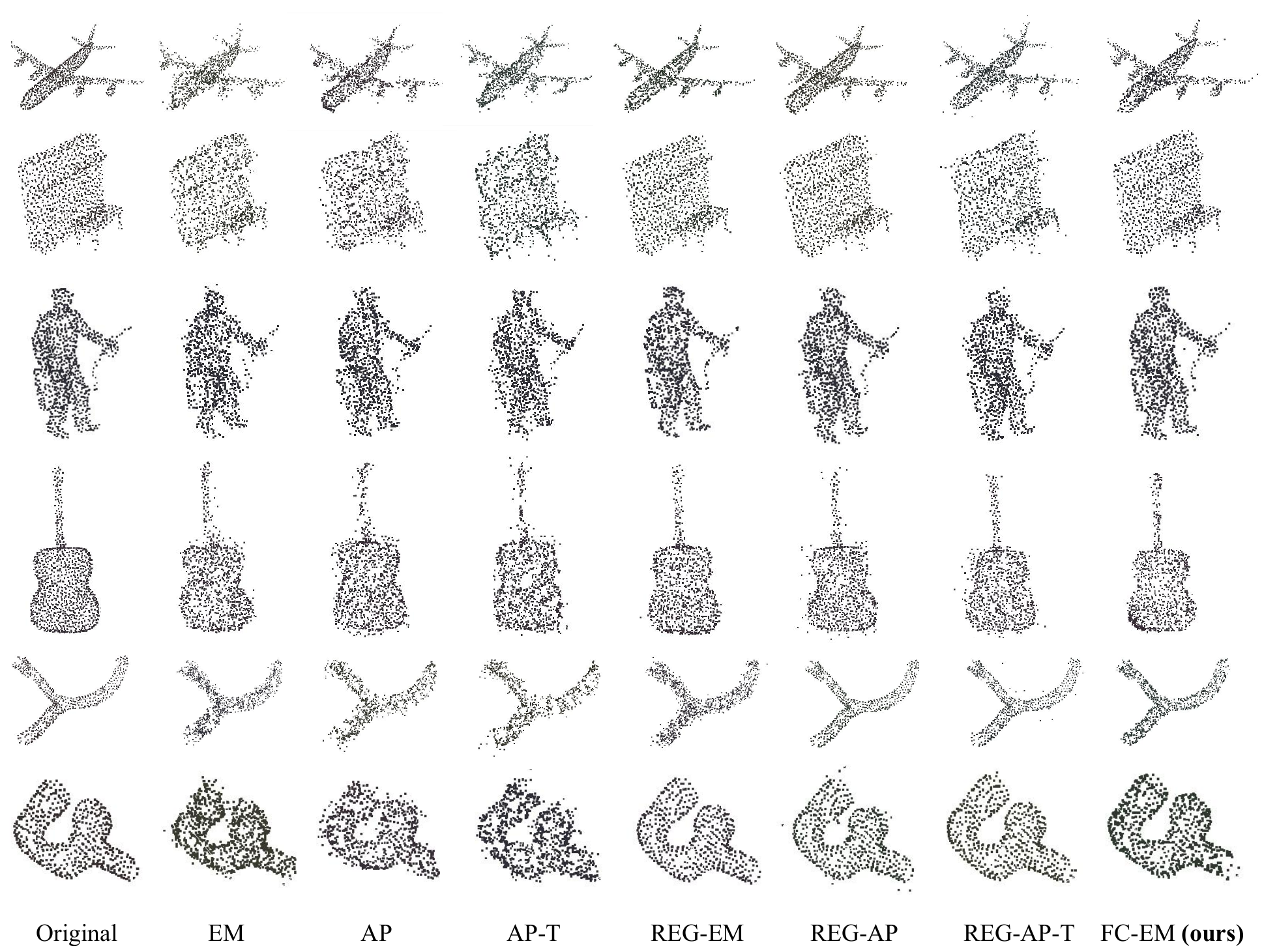}
\vspace{-1ex}
\caption{
More qualitative visualization results of baseline methods and our FC-EM.
Our FC-EM approach demonstrates enhanced naturalness and imperceptibility.}
\label{fig:fig-more}
\vspace{-3mm}
\end{figure*}

\section{Conclusion}
\label{sec-conc}
In this paper, we propose a Feature Collision Error-Minimization (FC-EM) method for effective 3D availability attacks to safeguard data privacy.
Theoretically, directly extending 2D availability attacks to 3D under distance regularization terms can lead to degeneracy, resulting in weaker poisons, 
due to optimization directions out of control in the bi-level optimization.
FC-EM establishes additional shortcuts in the feature space, inducing diverse update directions to prevent degeneracy.
We conduct a theoretical analysis of the effectiveness of FC-EM attack,
demonstrating that FC-EM exhibits stronger linear separability, consequently possessing stronger poison. Extensive experiments consistently validate the superiority and practicality of FC-EM.

\vspace{-1ex}
\paragraph{Limitations and future work.} 
While availability attacks are intended to prevent the unauthorized use of data, malicious entities can exploit availability attacks to compromise 3D deep models.
Therefore, defensive strategy against 3D availability attacks need to be developed to safeguard models from corruption.
Although our FC-EM method demonstrates strong attack ability and decent imperceptibility, it still exists some sparsity disparity, with some clustering phenomena as observed in Figure \ref{fig:fig5_1_2}.
Consequently, addressing the challenge of availability attacks on 3D point clouds with enhanced imperceptibility remains a promising direction for future research.

\vspace{-1ex}
\section*{Impact Statement}
\vspace{-0.5ex}
Data privacy and security issues in 3D deep learning is a severe problem towards safe and reliable 3D perception. Our work proposes an effective method to solve this issue, aiming to thwart unauthorized deep models from illegitimately learning data, which does not raise negative social impact. We aspire for our work to serve as a baseline for advancing research in the protection of privacy data within the 3D domain.

\vspace{-1ex}
\section*{Acknowledgement}
\vspace{-0.5ex}
This work is supported by NSFC grant (Nos. 12288201, 62276149)
and NKRDP grant No.2018YFA0306702.
Yinpeng Dong is also supported by the China National Postdoctoral Program for Innovative Talents.

\newpage

\bibliography{refpoi}

\begin{thebibliography}{79}
\providecommand{\natexlab}[1]{#1}
\providecommand{\url}[1]{\texttt{#1}}
\expandafter\ifx\csname urlstyle\endcsname\relax
  \providecommand{\doi}[1]{doi: #1}\else
  \providecommand{\doi}{doi: \begingroup \urlstyle{rm}\Url}\fi

\bibitem[Abadi et~al.(2016)Abadi, Chu, Goodfellow, McMahan, Mironov, Talwar, and Zhang]{abadi2016deep}
Abadi, M., Chu, A., Goodfellow, I., McMahan, H.~B., Mironov, I., Talwar, K., and Zhang, L.
\newblock Deep learning with differential privacy.
\newblock In \emph{Proceedings of the 2016 ACM SIGSAC conference on computer and communications security}, pp.\  308--318, 2016.

\bibitem[Burt(2020)]{burt2020facial}
Burt, C.
\newblock Facial biometrics training dataset leads to bipa lawsuits against amazon, alphabet and microsoft, 2020.
\newblock URL \url{https://reurl.cc/dV4rD8}.

\bibitem[Cao et~al.(2019)Cao, Xiao, Yang, Fang, Yang, Liu, and Li]{cao2019adversarial}
Cao, Y., Xiao, C., Yang, D., Fang, J., Yang, R., Liu, M., and Li, B.
\newblock Adversarial objects against lidar-based autonomous driving systems.
\newblock \emph{arXiv preprint arXiv:1907.05418}, 2019.

\bibitem[Carlini \& Wagner(2017)Carlini and Wagner]{Carlini-Wagner2017}
Carlini, N. and Wagner, D.
\newblock Towards evaluating the robustness of neural networks.
\newblock In \emph{IEEE Symposium on Security \& Privacy}, pp.\  39--57. IEEE, 2017.

\bibitem[Chen et~al.(2023)Chen, Yuan, Cheng, Gong, Qin, Wang, and Huang]{chen2023self}
Chen, S., Yuan, G., Cheng, X., Gong, Y., Qin, M., Wang, Y., and Huang, X.
\newblock Self-ensemble protection: Training checkpoints are good data protectors.
\newblock In \emph{The Eleventh International Conference on Learning Representations}, 2023.

\bibitem[Chen et~al.(2020)Chen, Kornblith, Norouzi, and Hinton]{chen2020simple}
Chen, T., Kornblith, S., Norouzi, M., and Hinton, G.
\newblock A simple framework for contrastive learning of visual representations.
\newblock In \emph{International conference on machine learning}, pp.\  1597--1607. PMLR, 2020.

\bibitem[Chen et~al.(2017{\natexlab{a}})Chen, Liu, Li, Lu, and Song]{chen2017targeted}
Chen, X., Liu, C., Li, B., Lu, K., and Song, D.
\newblock Targeted backdoor attacks on deep learning systems using data poisoning.
\newblock \emph{arXiv preprint arXiv:1712.05526}, 2017{\natexlab{a}}.

\bibitem[Chen et~al.(2017{\natexlab{b}})Chen, Ma, Wan, Li, and Xia]{chen2017multi}
Chen, X., Ma, H., Wan, J., Li, B., and Xia, T.
\newblock Multi-view 3d object detection network for autonomous driving.
\newblock In \emph{Proceedings of the IEEE conference on Computer Vision and Pattern Recognition}, pp.\  1907--1915, 2017{\natexlab{b}}.

\bibitem[Conklin(2020)]{conklin2020facebook}
Conklin, A.
\newblock Facebook agrees to pay record \$650m to settle facial recognition lawsuit, 2020.
\newblock URL \url{https://reurl.cc/Gd2kev}.

\bibitem[Das(2023)]{das2023healthdata}
Das, S.
\newblock Private uk health data donated for medical research shared with insurance companies, 2023.
\newblock URL \url{https://reurl.cc/WRODdy}.

\bibitem[Dolatabadi et~al.(2023)Dolatabadi, Erfani, and Leckie]{dolatabadi2023devil}
Dolatabadi, H.~M., Erfani, S., and Leckie, C.
\newblock The devil's advocate: Shattering the illusion of unexploitable data using diffusion models.
\newblock \emph{arXiv preprint arXiv:2303.08500}, 2023.

\bibitem[Fan et~al.(2017)Fan, Su, and Guibas]{fan2017point}
Fan, H., Su, H., and Guibas, L.~J.
\newblock A point set generation network for 3d object reconstruction from a single image.
\newblock In \emph{Proceedings of the IEEE conference on computer vision and pattern recognition}, pp.\  605--613, 2017.

\bibitem[Feng et~al.(2019)Feng, Cai, and Zhou]{feng2019learning}
Feng, J., Cai, Q.-Z., and Zhou, Z.-H.
\newblock Learning to confuse: generating training time adversarial data with auto-encoder.
\newblock \emph{Advances in Neural Information Processing Systems}, 32, 2019.

\bibitem[Fowl et~al.(2021)Fowl, Goldblum, Chiang, Geiping, Czaja, and Goldstein]{fowl2021adversarial}
Fowl, L., Goldblum, M., Chiang, P.-y., Geiping, J., Czaja, W., and Goldstein, T.
\newblock Adversarial examples make strong poisons.
\newblock \emph{Advances in Neural Information Processing Systems}, 34:\penalty0 30339--30351, 2021.

\bibitem[Fu et~al.(2022)Fu, He, Liu, Shen, and Tao]{Fu2022RobustUE}
Fu, S., He, F., Liu, Y., Shen, L., and Tao, D.
\newblock Robust unlearnable examples: Protecting data against adversarial learning.
\newblock In \emph{International Conference on Learning Representations}, 2022.

\bibitem[Gao et~al.(2023)Gao, Bai, Wu, Ya, and Xia]{gao2023imperceptible}
Gao, K., Bai, J., Wu, B., Ya, M., and Xia, S.-T.
\newblock Imperceptible and robust backdoor attack in 3d point cloud.
\newblock \emph{IEEE Transactions on Information Forensics and Security}, 19:\penalty0 1267--1282, 2023.

\bibitem[Gerig et~al.(2018)Gerig, Morel-Forster, Blumer, Egger, Luthi, Sch{\"o}nborn, and Vetter]{gerig2018morphable}
Gerig, T., Morel-Forster, A., Blumer, C., Egger, B., Luthi, M., Sch{\"o}nborn, S., and Vetter, T.
\newblock Morphable face models-an open framework.
\newblock In \emph{2018 13th IEEE International Conference on Automatic Face \& Gesture Recognition (FG 2018)}, pp.\  75--82. IEEE, 2018.

\bibitem[He et~al.(2023{\natexlab{a}})He, Liu, Li, Liang, Li, Jia, and Cao]{he2023generating}
He, B., Liu, J., Li, Y., Liang, S., Li, J., Jia, X., and Cao, X.
\newblock Generating transferable 3d adversarial point cloud via random perturbation factorization.
\newblock In \emph{Proceedings of the AAAI Conference on Artificial Intelligence}, volume~37, pp.\  764--772, 2023{\natexlab{a}}.

\bibitem[He et~al.(2022)He, Zha, and Katabi]{he2022indiscriminate}
He, H., Zha, K., and Katabi, D.
\newblock Indiscriminate poisoning attacks on unsupervised contrastive learning.
\newblock In \emph{The Eleventh International Conference on Learning Representations}, 2022.

\bibitem[He et~al.(2023{\natexlab{b}})He, Xu, Ren, Cui, Liu, Aggarwal, and Tang]{he2023sharpness}
He, P., Xu, H., Ren, J., Cui, Y., Liu, H., Aggarwal, C.~C., and Tang, J.
\newblock Sharpness-aware data poisoning attack.
\newblock \emph{arXiv preprint arXiv:2305.14851}, 2023{\natexlab{b}}.

\bibitem[Huang et~al.(2020)Huang, Ma, Erfani, Bailey, and Wang]{huang2020unlearnable}
Huang, H., Ma, X., Erfani, S.~M., Bailey, J., and Wang, Y.
\newblock Unlearnable examples: Making personal data unexploitable.
\newblock In \emph{International Conference on Learning Representations}, 2020.

\bibitem[Huang et~al.(2022)Huang, Dong, Chen, Zhou, Zhang, and Yu]{huang2022shape}
Huang, Q., Dong, X., Chen, D., Zhou, H., Zhang, W., and Yu, N.
\newblock Shape-invariant 3d adversarial point clouds.
\newblock In \emph{Proceedings of the IEEE/CVF Conference on Computer Vision and Pattern Recognition}, pp.\  15335--15344, 2022.

\bibitem[Jiang et~al.(2023)Jiang, Diao, Wang, Sun, Wang, and Hong]{jiang2023unlearnable}
Jiang, W., Diao, Y., Wang, H., Sun, J., Wang, M., and Hong, R.
\newblock Unlearnable examples give a false sense of security: Piercing through unexploitable data with learnable examples.
\newblock \emph{arXiv preprint arXiv:2305.09241}, 2023.

\bibitem[Lee et~al.(2021)Lee, Lee, Lee, Lee, Lee, Woo, and Lee]{lee2021regularization}
Lee, D., Lee, J., Lee, J., Lee, H., Lee, M., Woo, S., and Lee, S.
\newblock Regularization strategy for point cloud via rigidly mixed sample.
\newblock In \emph{Proceedings of the IEEE/CVF Conference on Computer Vision and Pattern Recognition}, pp.\  15900--15909, 2021.

\bibitem[Li et~al.(2021)Li, Chen, Zhao, Tong, Zhao, Lim, and Zhou]{li2021pointba}
Li, X., Chen, Z., Zhao, Y., Tong, Z., Zhao, Y., Lim, A., and Zhou, J.~T.
\newblock Pointba: Towards backdoor attacks in 3d point cloud.
\newblock In \emph{Proceedings of the IEEE/CVF International Conference on Computer Vision}, pp.\  16492--16501, 2021.

\bibitem[Li et~al.(2023)Li, Liu, and Gao]{li2023make}
Li, X., Liu, M., and Gao, S.
\newblock Make text unlearnable: Exploiting effective patterns to protect personal data.
\newblock In \emph{The Third Workshop on Trustworthy Natural Language Processing}, pp.\  249, 2023.

\bibitem[Liu \& Hu(2022)Liu and Hu]{liu2022imperceptible}
Liu, D. and Hu, W.
\newblock Imperceptible transfer attack and defense on 3d point cloud classification.
\newblock \emph{IEEE Transactions on Pattern Analysis and Machine Intelligence}, 2022.

\bibitem[Liu et~al.(2023)Liu, Zhao, and Larson]{liu2023image}
Liu, Z., Zhao, Z., and Larson, M.
\newblock Image shortcut squeezing: Countering perturbative availability poisons with compression.
\newblock In \emph{International conference on machine learning}, 2023.

\bibitem[Lou et~al.(2024)Lou, Jia, Gu, Liu, Liang, He, and Cao]{lou2024hide}
Lou, T., Jia, X., Gu, J., Liu, L., Liang, S., He, B., and Cao, X.
\newblock Hide in thicket: Generating imperceptible and rational adversarial perturbations on 3d point clouds.
\newblock \emph{arXiv preprint arXiv:2403.05247}, 2024.

\bibitem[Ma et~al.(2022)Ma, Qin, You, Ran, and Fu]{ma2022rethinking}
Ma, X., Qin, C., You, H., Ran, H., and Fu, Y.
\newblock Rethinking network design and local geometry in point cloud: A simple residual mlp framework.
\newblock \emph{arXiv preprint arXiv:2202.07123}, 2022.

\bibitem[Madry et~al.(2018)Madry, Makelov, Schmidt, Tsipras, and Vladu]{madry2018towards}
Madry, A., Makelov, A., Schmidt, L., Tsipras, D., and Vladu, A.
\newblock Towards deep learning models resistant to adversarial attacks.
\newblock In \emph{International Conference on Learning Representations}, 2018.

\bibitem[Malihi et~al.(2016)Malihi, Valadan~Zoej, Hahn, Mokhtarzade, and Arefi]{malihi20163d}
Malihi, S., Valadan~Zoej, M., Hahn, M., Mokhtarzade, M., and Arefi, H.
\newblock 3d building reconstruction using dense photogrammetric point cloud.
\newblock \emph{The International Archives of the Photogrammetry, Remote Sensing and Spatial Information Sciences}, 41:\penalty0 71--74, 2016.

\bibitem[Miao et~al.(2022)Miao, Dong, Zhu, and Gao]{miao2022isometric}
Miao, Y., Dong, Y., Zhu, J., and Gao, X.-S.
\newblock Isometric 3d adversarial examples in the physical world.
\newblock In \emph{Advances in Neural Information Processing Systems}, 2022.

\bibitem[Qi et~al.(2017{\natexlab{a}})Qi, Su, Mo, and Guibas]{qi2017pointnet}
Qi, C.~R., Su, H., Mo, K., and Guibas, L.~J.
\newblock Pointnet: Deep learning on point sets for 3d classification and segmentation.
\newblock In \emph{Proceedings of the IEEE conference on computer vision and pattern recognition}, pp.\  652--660, 2017{\natexlab{a}}.

\bibitem[Qi et~al.(2017{\natexlab{b}})Qi, Yi, Su, and Guibas]{qi2017pointnet++}
Qi, C.~R., Yi, L., Su, H., and Guibas, L.~J.
\newblock Pointnet++: Deep hierarchical feature learning on point sets in a metric space.
\newblock \emph{Advances in neural information processing systems}, 30, 2017{\natexlab{b}}.

\bibitem[Qin et~al.(2023)Qin, Gao, Zhao, Ye, and Xu]{qin2023learning}
Qin, T., Gao, X., Zhao, J., Ye, K., and Xu, C.-Z.
\newblock Learning the unlearnable: Adversarial augmentations suppress unlearnable example attacks.
\newblock \emph{arXiv preprint arXiv:2303.15127}, 2023.

\bibitem[Ren et~al.(2022)Ren, Xu, Wan, Ma, Sun, and Tang]{ren2022transferable}
Ren, J., Xu, H., Wan, Y., Ma, X., Sun, L., and Tang, J.
\newblock Transferable unlearnable examples.
\newblock In \emph{The Eleventh International Conference on Learning Representations}, 2022.

\bibitem[Rockafellar \& Wets(2009)Rockafellar and Wets]{rockafellar2009variational}
Rockafellar, R.~T. and Wets, R. J.-B.
\newblock \emph{Variational analysis}, volume 317.
\newblock Springer Science \& Business Media, 2009.

\bibitem[Sadasivan et~al.(2023)Sadasivan, Soltanolkotabi, and Feizi]{sadasivan2023cuda}
Sadasivan, V.~S., Soltanolkotabi, M., and Feizi, S.
\newblock Cuda: Convolution-based unlearnable datasets.
\newblock In \emph{Proceedings of the IEEE/CVF Conference on Computer Vision and Pattern Recognition}, pp.\  3862--3871, 2023.

\bibitem[Sandoval-Segura et~al.(2022{\natexlab{a}})Sandoval-Segura, Singla, Fowl, Geiping, Goldblum, Jacobs, and Goldstein]{sandoval2022poisons}
Sandoval-Segura, P., Singla, V., Fowl, L., Geiping, J., Goldblum, M., Jacobs, D., and Goldstein, T.
\newblock Poisons that are learned faster are more effective.
\newblock In \emph{Proceedings of the IEEE/CVF Conference on Computer Vision and Pattern Recognition}, pp.\  198--205, 2022{\natexlab{a}}.

\bibitem[Sandoval-Segura et~al.(2022{\natexlab{b}})Sandoval-Segura, Singla, Geiping, Goldblum, Goldstein, and Jacobs]{sandoval2022autoregressive}
Sandoval-Segura, P., Singla, V., Geiping, J., Goldblum, M., Goldstein, T., and Jacobs, D.
\newblock Autoregressive perturbations for data poisoning.
\newblock \emph{Advances in Neural Information Processing Systems}, 35:\penalty0 27374--27386, 2022{\natexlab{b}}.

\bibitem[Sandoval-Segura et~al.(2023)Sandoval-Segura, Singla, Geiping, Goldblum, and Goldstein]{sandoval2023can}
Sandoval-Segura, P., Singla, V., Geiping, J., Goldblum, M., and Goldstein, T.
\newblock What can we learn from unlearnable datasets?
\newblock \emph{arXiv preprint arXiv:2305.19254}, 2023.

\bibitem[Shokri \& Shmatikov(2015)Shokri and Shmatikov]{shokri2015privacy}
Shokri, R. and Shmatikov, V.
\newblock Privacy-preserving deep learning.
\newblock In \emph{Proceedings of the 22nd ACM SIGSAC conference on computer and communications security}, pp.\  1310--1321, 2015.

\bibitem[Shokri et~al.(2017)Shokri, Stronati, Song, and Shmatikov]{shokri2017membership}
Shokri, R., Stronati, M., Song, C., and Shmatikov, V.
\newblock Membership inference attacks against machine learning models.
\newblock In \emph{2017 IEEE symposium on security and privacy (SP)}, pp.\  3--18. IEEE, 2017.

\bibitem[Taha \& Hanbury(2015)Taha and Hanbury]{taha2015metrics}
Taha, A.~A. and Hanbury, A.
\newblock Metrics for evaluating 3d medical image segmentation: analysis, selection, and tool.
\newblock \emph{BMC medical imaging}, 15\penalty0 (1):\penalty0 1--28, 2015.

\bibitem[Tao et~al.(2021)Tao, Feng, Yi, Huang, and Chen]{tao2021better}
Tao, L., Feng, L., Yi, J., Huang, S.-J., and Chen, S.
\newblock Better safe than sorry: Preventing delusive adversaries with adversarial training.
\newblock \emph{Advances in Neural Information Processing Systems}, 34:\penalty0 16209--16225, 2021.

\bibitem[Tao et~al.(2022)Tao, Feng, Wei, Yi, Huang, and Chen]{tao2022can}
Tao, L., Feng, L., Wei, H., Yi, J., Huang, S.-J., and Chen, S.
\newblock Can adversarial training be manipulated by non-robust features?
\newblock \emph{Advances in Neural Information Processing Systems}, 35:\penalty0 26504--26518, 2022.

\bibitem[Thanh-Tung \& Tran(2020)Thanh-Tung and Tran]{thanh2020catastrophic}
Thanh-Tung, H. and Tran, T.
\newblock Catastrophic forgetting and mode collapse in gans.
\newblock In \emph{2020 international joint conference on neural networks (ijcnn)}, pp.\  1--10. IEEE, 2020.

\bibitem[Tsai et~al.(2020)Tsai, Yang, Ho, and Jin]{tsai2020robust}
Tsai, T., Yang, K., Ho, T.-Y., and Jin, Y.
\newblock Robust adversarial objects against deep learning models.
\newblock In \emph{Proceedings of the AAAI Conference on Artificial Intelligence}, volume~34, pp.\  954--962, 2020.

\bibitem[Uy et~al.(2019)Uy, Pham, Hua, Nguyen, and Yeung]{uy2019revisiting}
Uy, M.~A., Pham, Q.-H., Hua, B.-S., Nguyen, T., and Yeung, S.-K.
\newblock Revisiting point cloud classification: A new benchmark dataset and classification model on real-world data.
\newblock In \emph{Proceedings of the IEEE/CVF international conference on computer vision}, pp.\  1588--1597, 2019.

\bibitem[Vincent(2019)]{vincent2019google}
Vincent, J.
\newblock Google accused of inappropriate access to medical data in potential class-action lawsuit, 2019.
\newblock URL \url{https://reurl.cc/bzK69v}.

\bibitem[Wang et~al.(2019)Wang, Sun, Liu, Sarma, Bronstein, and Solomon]{wang2019dynamic}
Wang, Y., Sun, Y., Liu, Z., Sarma, S.~E., Bronstein, M.~M., and Solomon, J.~M.
\newblock Dynamic graph cnn for learning on point clouds.
\newblock \emph{ACM Transactions on Graphics (tog)}, 38\penalty0 (5):\penalty0 1--12, 2019.

\bibitem[Wang et~al.(2024)Wang, Zhu, and Gao]{wang2024efficient}
Wang, Y., Zhu, Y., and Gao, X.-S.
\newblock Efficient availability attacks against supervised and contrastive learning simultaneously.
\newblock \emph{arXiv preprint arXiv:2402.04010}, 2024.

\bibitem[Wen et~al.(2023)Wen, Zhao, Liu, Backes, Wang, and Zhang]{wen2023is}
Wen, R., Zhao, Z., Liu, Z., Backes, M., Wang, T., and Zhang, Y.
\newblock Is adversarial training really a silver bullet for mitigating data poisoning?
\newblock In \emph{International Conference on Learning Representations}, 2023.

\bibitem[Wen et~al.(2020)Wen, Lin, Chen, Chen, and Jia]{wen2020geometry}
Wen, Y., Lin, J., Chen, K., Chen, C.~P., and Jia, K.
\newblock Geometry-aware generation of adversarial point clouds.
\newblock \emph{IEEE Transactions on Pattern Analysis and Machine Intelligence}, 2020.

\bibitem[Wu et~al.(2022)Wu, Chen, Xie, and Huang]{wu2022one}
Wu, S., Chen, S., Xie, C., and Huang, X.
\newblock One-pixel shortcut: On the learning preference of deep neural networks.
\newblock In \emph{The Eleventh International Conference on Learning Representations}, 2022.

\bibitem[Wu et~al.(2023)Wu, Zhang, Fu, Wang, Ren, Pan, Wu, Yang, Wang, Qian, et~al.]{wu2023omniobject3d}
Wu, T., Zhang, J., Fu, X., Wang, Y., Ren, J., Pan, L., Wu, W., Yang, L., Wang, J., Qian, C., et~al.
\newblock Omniobject3d: Large-vocabulary 3d object dataset for realistic perception, reconstruction and generation.
\newblock In \emph{Proceedings of the IEEE/CVF Conference on Computer Vision and Pattern Recognition}, pp.\  803--814, 2023.

\bibitem[Wu et~al.(2015)Wu, Song, Khosla, Yu, Zhang, Tang, and Xiao]{wu20153d}
Wu, Z., Song, S., Khosla, A., Yu, F., Zhang, L., Tang, X., and Xiao, J.
\newblock 3d shapenets: A deep representation for volumetric shapes.
\newblock In \emph{Proceedings of the IEEE conference on computer vision and pattern recognition}, pp.\  1912--1920, 2015.

\bibitem[Wu et~al.(2020)Wu, Duan, Wang, Fan, and Guibas]{wu2020if}
Wu, Z., Duan, Y., Wang, H., Fan, Q., and Guibas, L.~J.
\newblock If-defense: 3d adversarial point cloud defense via implicit function based restoration.
\newblock \emph{arXiv preprint arXiv:2010.05272}, 2020.

\bibitem[Xiang et~al.(2019)Xiang, Qi, and Li]{xiang2019generating}
Xiang, C., Qi, C.~R., and Li, B.
\newblock Generating 3d adversarial point clouds.
\newblock In \emph{Proceedings of the IEEE/CVF Conference on Computer Vision and Pattern Recognition}, pp.\  9136--9144, 2019.

\bibitem[Xiang et~al.(2021)Xiang, Miller, Chen, Li, and Kesidis]{xiang2021backdoor}
Xiang, Z., Miller, D.~J., Chen, S., Li, X., and Kesidis, G.
\newblock A backdoor attack against 3d point cloud classifiers.
\newblock In \emph{Proceedings of the IEEE/CVF International Conference on Computer Vision}, pp.\  7597--7607, 2021.

\bibitem[Xu et~al.(2021{\natexlab{a}})Xu, Ding, Zhao, and Qi]{xu2021paconv}
Xu, M., Ding, R., Zhao, H., and Qi, X.
\newblock Paconv: Position adaptive convolution with dynamic kernel assembling on point clouds.
\newblock In \emph{Proceedings of the IEEE/CVF Conference on Computer Vision and Pattern Recognition}, pp.\  3173--3182, 2021{\natexlab{a}}.

\bibitem[Xu et~al.(2021{\natexlab{b}})Xu, Zhang, Zhou, Xu, Qi, and Qiao]{xu2021learning}
Xu, M., Zhang, J., Zhou, Z., Xu, M., Qi, X., and Qiao, Y.
\newblock Learning geometry-disentangled representation for complementary understanding of 3d object point cloud.
\newblock In \emph{Proceedings of the AAAI conference on artificial intelligence}, volume~35, pp.\  3056--3064, 2021{\natexlab{b}}.

\bibitem[Xu et~al.(2023)Xu, Wang, Wang, Chen, Pang, and Lin]{xu2023pointllm}
Xu, R., Wang, X., Wang, T., Chen, Y., Pang, J., and Lin, D.
\newblock Pointllm: Empowering large language models to understand point clouds.
\newblock \emph{arXiv preprint arXiv:2308.16911}, 2023.

\bibitem[Yang et~al.(2020)Yang, Xia, Kin, and Igarashi]{yang2020intra}
Yang, X., Xia, D., Kin, T., and Igarashi, T.
\newblock Intra: 3d intracranial aneurysm dataset for deep learning.
\newblock In \emph{Proceedings of the IEEE/CVF Conference on Computer Vision and Pattern Recognition}, pp.\  2656--2666, 2020.

\bibitem[Yu et~al.(2022)Yu, Zhang, Chen, Yin, and Liu]{yu2022availability}
Yu, D., Zhang, H., Chen, W., Yin, J., and Liu, T.-Y.
\newblock Availability attacks create shortcuts.
\newblock In \emph{Proceedings of the 28th ACM SIGKDD Conference on Knowledge Discovery and Data Mining}, pp.\  2367--2376, 2022.

\bibitem[Yuan \& Wu(2021)Yuan and Wu]{yuan2021neural}
Yuan, C.-H. and Wu, S.-H.
\newblock Neural tangent generalization attacks.
\newblock In \emph{International Conference on Machine Learning}, pp.\  12230--12240. PMLR, 2021.

\bibitem[Yue et~al.(2018)Yue, Wu, Seshia, Keutzer, and Sangiovanni-Vincentelli]{yue2018lidar}
Yue, X., Wu, B., Seshia, S.~A., Keutzer, K., and Sangiovanni-Vincentelli, A.~L.
\newblock A lidar point cloud generator: from a virtual world to autonomous driving.
\newblock In \emph{Proceedings of the 2018 ACM on International Conference on Multimedia Retrieval}, pp.\  458--464, 2018.

\bibitem[Zhang et~al.(2018)Zhang, Cisse, Dauphin, and Lopez-Paz]{zhang2018mixup}
Zhang, H., Cisse, M., Dauphin, Y.~N., and Lopez-Paz, D.
\newblock mixup: Beyond empirical risk minimization.
\newblock In \emph{International Conference on Learning Representations}, 2018.

\bibitem[Zhang et~al.(2019)Zhang, Yu, Jiao, Xing, El~Ghaoui, and Jordan]{zhang2019theoretically}
Zhang, H., Yu, Y., Jiao, J., Xing, E., El~Ghaoui, L., and Jordan, M.
\newblock Theoretically principled trade-off between robustness and accuracy.
\newblock In \emph{International conference on machine learning}, pp.\  7472--7482. PMLR, 2019.

\bibitem[Zhang et~al.(2022{\natexlab{a}})Zhang, Chen, Ouyang, Liu, Zhu, Chen, Meng, and Wu]{zhang2022pointcutmix}
Zhang, J., Chen, L., Ouyang, B., Liu, B., Zhu, J., Chen, Y., Meng, Y., and Wu, D.
\newblock Pointcutmix: Regularization strategy for point cloud classification.
\newblock \emph{Neurocomputing}, 505:\penalty0 58--67, 2022{\natexlab{a}}.

\bibitem[Zhang et~al.(2023{\natexlab{a}})Zhang, Chen, Liu, Ouyang, Xie, Zhu, Li, and Meng]{zhang20233d}
Zhang, J., Chen, L., Liu, B., Ouyang, B., Xie, Q., Zhu, J., Li, W., and Meng, Y.
\newblock 3d adversarial attacks beyond point cloud.
\newblock \emph{Information Sciences}, 633:\penalty0 491--503, 2023{\natexlab{a}}.

\bibitem[Zhang et~al.(2023{\natexlab{b}})Zhang, Ma, Yi, Sang, Jiang, Wang, and Xu]{zhang2023unlearnable}
Zhang, J., Ma, X., Yi, Q., Sang, J., Jiang, Y.-G., Wang, Y., and Xu, C.
\newblock Unlearnable clusters: Towards label-agnostic unlearnable examples.
\newblock In \emph{Proceedings of the IEEE/CVF Conference on Computer Vision and Pattern Recognition}, pp.\  3984--3993, 2023{\natexlab{b}}.

\bibitem[Zhang et~al.(2022{\natexlab{b}})Zhang, Da, and Yu]{zhang2022learning}
Zhang, Z., Da, F., and Yu, Y.
\newblock Learning directly from synthetic point clouds for “in-the-wild” 3d face recognition.
\newblock \emph{Pattern Recognition}, 123:\penalty0 108394, 2022{\natexlab{b}}.

\bibitem[Zhao et~al.(2021)Zhao, Jiang, Jia, Torr, and Koltun]{zhao2021point}
Zhao, H., Jiang, L., Jia, J., Torr, P.~H., and Koltun, V.
\newblock Point transformer.
\newblock In \emph{Proceedings of the IEEE/CVF international conference on computer vision}, pp.\  16259--16268, 2021.

\bibitem[Zheng et~al.(2019)Zheng, Chen, Yuan, Li, and Ren]{zheng2019pointcloud}
Zheng, T., Chen, C., Yuan, J., Li, B., and Ren, K.
\newblock Pointcloud saliency maps.
\newblock In \emph{Proceedings of the IEEE/CVF International Conference on Computer Vision}, pp.\  1598--1606, 2019.

\bibitem[Zhou et~al.(2023)Zhou, Wang, Ma, Liu, Huang, and Wang]{zhou2023uni3d}
Zhou, J., Wang, J., Ma, B., Liu, Y.-S., Huang, T., and Wang, X.
\newblock Uni3d: Exploring unified 3d representation at scale.
\newblock \emph{arXiv preprint arXiv:2310.06773}, 2023.

\bibitem[Zhou \& Xiao(2018)Zhou and Xiao]{zhou20183d}
Zhou, S. and Xiao, S.
\newblock 3d face recognition: a survey.
\newblock \emph{Human-centric Computing and Information Sciences}, 8\penalty0 (1):\penalty0 1--27, 2018.

\bibitem[Zhu et~al.(2023)Zhu, Yu, and Gao]{zhu2023detection}
Zhu, Y., Yu, L., and Gao, X.-S.
\newblock Detection and defense of unlearnable examples.
\newblock \emph{arXiv preprint arXiv:2312.08898}, 2023.

\end{thebibliography}
\bibliographystyle{icml2024}

\appendix
\onecolumn

\section{Theorems and Proofs}
\label{proofs}

\begin{theorem}[Restate of Theorem \ref{th-new-loss-2}]
\label{th-new-loss-2-2}
Let $D=\{(x_i, y_i)\}_{i=1}^N$, $\delta_i$ be the perturbations corresponding to $x_i$.
Denote the equilibrium for bi-level optimization $\min_{\theta,\delta}\{\L_1, \L_2\}$ as $(\theta^*, \delta^*)$.
Assume that $\L_{\textup{new}}$ is $L$-lipschitz with respect to $\delta$ under $\ell_2$ norm, $\L_{\dist}$ is Chamfer distance,  and for each point $x_i^j$, the $\ell_2$ norm of perturbation $\delta_i^j$ is bound by $\frac{1}{2}\min_{j,k}\lVert x_i^j-x_i^k \rVert_2$. 
Then it holds $\lVert \delta_i^* \rVert_2 \geq \frac{1}{4\beta}(L-\sqrt{L^2-8\beta\Delta})$, where $\Delta=\L_2(\theta^*, 0) - \L_2(\theta^*, \delta^*)$.
\end{theorem}
\begin{proof}
For any $\delta$, due to the L-lipschitz property, it has
\begin{align*}
\L_2(\theta^*,\delta)-\L_2(\theta^*,0) & =\frac{1}{N}\sum\limits_{j=1}^N\left[\L_{\new}(x_i+\delta_i,y_i,\theta^*)-\L_{\new}(x_i+0,y_i,\theta^*)\right]+ \frac{\beta}{N}\sum\limits_{j=1}^N \L_{\dist}(x_i+\delta_i,x_i) \\
& \geq -\frac{L}{N}\sum\limits_{j=1}^N \|\delta_i\|_2 + \frac{\beta}{N}\sum\limits_{j=1}^N \L_{\dist}(x_i+\delta_i,x_i).
\end{align*}
Because for each point cloud $x_i^j$ of $x_i$, $\lVert \delta_i^j \rVert_2 \leq \frac{1}{2}\min_{j,k}\|x_i^j-x_i^k\|_2$, 
it holds that 
$$\L_{\dist}(x_i+\delta_i,x_i)=2\lVert\delta_i\rVert_2^2.$$
Then there must exists $i\in[N]$, such that 
 $$\L_2(\theta^*,\delta) - \L_2(\theta^*,\delta^*) \geq -L\lVert\delta_i\rVert_2 + \Delta + 2\beta\lVert\delta_i\rVert_2^2.$$
Let $\delta = \delta^*$, one must have
$$L\lVert\delta_i^*\rVert_2 - 2\beta\lVert\delta_i^*\rVert_2^2 \geq \Delta .$$
To make the above inequality satisfy, the necessary condition for $\delta_i^*$ is  $$\lVert\delta_i^*\rVert_2 \geq \frac{L-\sqrt{L^2-8\beta\Delta}}{4\beta}.$$
\end{proof}

\begin{theorem}[Restate of Theorem \ref{th-linear}]
\label{th-linear-2}
Let 
 $\alpha$ be the loss minimum of dataset $D$ under all linear classifier, the weight matrix of a linear classifier be denoted as $W$, where $\ell_2$ norm of each row is normalized to $\sqrt{d}$. For any two rows of $W$, their cosine similarity does not exceed $\gamma$. The loss function is $\L_{\cls}(x,y;\theta)=\max( \max_{t\neq y}f_{\theta}(x)_t - f_{\theta}(x)_y, 0)$. 
Then, there exist poisons $\{\delta_i\}_{i=1}^N$ satisfying  
$\D_{\rm{c}}(D, D_{\delta})\leq O(\frac{\alpha}{(1-\gamma)^2})$,
such that the poisoned dataset is linearly separable.
\end{theorem}
\begin{proof}
Assume that $\delta_i=\beta_i W_{y_i}$.
We only need to prove that there exists weight matrix $W=[W_1^T,\cdots,W_C^T]$, such that $W_{y_i}(x_i + \delta_i) > W_{y'}(x_i + \delta_i)$ holds for $y'\neq y_i$ for all $i\in[N]$.

This can be concluded from 
$$\beta_i(\lVert W_{y_i}\rVert_2^2-\lVert W_{y_i}\rVert_2\lVert W_{y'}\rVert_2 s(W_{y_i},W_{y'})) > \max\limits_{y'\neq y_i}(W_{y'}x_i-W_{y_i}x_i), \forall i,$$
which can be derived from 
$$\beta_i > \frac{\L_{\cls}(x_i,y_i,W)}{\lVert W_{y_i}\rVert_2(\lVert W_{y_i}\rVert_2-\gamma\|W_{y'}\|_2)}, \forall i.$$
Therefore, as long as $\delta_i$ satisfies 
$$\lVert\delta_i\rVert_2>\frac{\L_{\cls}(x_i,y_i,W)}{d(1-\gamma)}, \forall i,$$
the above inequality satisfies. 

The chamfer distance of clean data and poisoned data satisfies that
$$\D_c(D, {\delta})\leq\frac{2}{N}\sum\limits_{i=1}^N\lVert\delta_i\rVert_2^2,$$
meanwhile, we can find $\lVert\delta_i\rVert_2\leq \frac{2 Loss(x_i,y_i,W)}{d(1-\gamma)}$ satisfies abve inequalities.
Due to normalization, $$\L_{\cls}(x_i,y_i,W)\leq\max\limits_{y'\neq y_i}(W_{y'}x_i-W_{y_i}x_i)\leq 2\sqrt{d}\lVert x_i\rVert_2\leq 2d,$$ 
as $ x_i$ lies in $[-1,1]$. 

Therefore, 
$$\frac{2}{N}\sum\limits_{i=1}^N\lVert\delta_i\lVert_2^2\leq \frac{8}{(1-\gamma)^2 d^2}\frac{1}{N}\sum\limits_{i=1}^N \L_{\cls}^2(x_i, y_i, W) \leq \frac{32}{(1-\gamma)^2}\frac{1}{N}\sum\limits_{i=1}^N \L_{\cls}(x_i, y_i, W).$$
The above inequality holds for all $W$, hence we only need  $\D_c(D, D_{\delta})\leq O(\frac{\alpha}{(1-\gamma)^2})$.
\end{proof}

\begin{theorem}[Restate of Theorem \ref{err-min-th}]
\label{err-min-th-2}
Assume that $\L_{\cls}$ and $\L_{\dist}$ are continuous, and the network's hypothesis space $\H_{\F}$ is compact. Let $D_{\delta}= \{(x_i+\delta_i, y_i)\}_{i=1}^N$ be the poisoned dataset of $D$. 
Then, there exists an optimal point ($\delta^*$, $\theta^*$) for the bi-level optimization \eqref{err-min-equ}, 
and the optimal perturbation $\delta^*$ satisfies $\L_{\dist}(D, \delta^*)\leq \frac{1}{\beta}\big[\min_{\theta} \L_{\cls}(D;\theta) -\min_{\delta, \theta} \L_{\cls}(D_{\delta};\theta)\big]$.
\end{theorem}
\begin{proof}
Since $\L_{\cls}$ and $\L_{\dist}$ are continuous, their sum and expectation are also continuous. Therefore, since $\H_{\F}$ is compact, the optimal minimum point $\delta^*$ exists.
In addition, as $\delta\in \R^{n\times 3}$ is the Euclid space, we may further assume that $\delta\in [-2, 2]^{n\times 3}$ as we can assume 3D point clouds always lie in $ [-1, 1]^{n\times 3}$ after normalization. 
Therefore, the perturbation space is also compact, resulting in the optimal minimum point $\delta^*$ exists.

Denote 
$$\alpha_1 = \min\limits_{\theta} \L_{\cls}(D;\theta) = \min\limits_{\theta} \E_{(x, y) \in D} \L_{\cls}(x, y;\theta),$$
$$\alpha_2 = \min\limits_{\delta, \theta} \L_{\cls}(D_{\delta};\theta) = \min\limits_{\theta}\min\limits_{\delta} \E_{(x, y) \in D} \L_{\cls}(x+\delta(x), y;\theta).$$
There exists a $\theta_0$ such that $\L_{\cls}(D;\theta_0)=\alpha_1$. 
Therefore, the optimal point ($\delta^*$, $\theta^*$) satisfies 
$$\E_{(x, y) \in D }\big[ \L_{\cls}(x+\delta^*,y;\theta^*) + \beta \cdot \L_{\dist} (x+\delta^*, x)\big]\leq  \E_{(x, y) \in D }\big[ \L_{\cls}(x+0,y;\theta_0) + \beta \cdot \L_{\dist} (x+0, x)\big]=\alpha_1.$$
Since $\L_{\cls}$ is always not less than $\alpha_2$, $\delta^*$ is in $\{\delta\sep \L_{\dist}(D, \delta)\leq \frac{\alpha_1 - \alpha_2}{\beta} \}$.
\end{proof}

\begin{corollary}[Restate of Corollary \ref{err-min-cor}]
\label{err-min-cor-2}
Let $\L_{\cls}(x,y;\theta)=\max( \max_{t\neq y}f_{\theta}(x)_t - f_{\theta}(x)_y, 0)$, where $f_{\theta}(\cdot)$ is the output of logits layer, and $\L_{\dist}$ be a metric function. If  $\H_{\F}$ includes the function capable of correctly classifying $D$, then the equilibrium of the bi-level optimization \eqref{err-min-equ} will degenerate to $\delta^*=0$.
\end{corollary}
\begin{proof}
%
Since $\H_{\F}$ contains the true function, it has 

$$\min\limits_{\theta}  \L_{\cls}(D;\theta)=0, \min\limits_{\delta,\theta}  \L_{\cls}(D_{\delta};\theta)=0.$$ 
Therefore,  by Theorem \ref{err-min-th}, as $\L_{\dist}$ is a metric function, $\delta^*$ lies in the perturbation space 
$$\{\delta\sep \L_{\dist}(D, \delta)\leq 0 \} = \{0\}.$$ 
\end{proof}

\section{More Discussions on Related Work}
\label{app-rel}
\paragraph{Availability attacks.}
Privacy issues have been extensively studied in the field of privacy-preserving machine learning \cite{shokri2015privacy, abadi2016deep, shokri2017membership}, including studies on availability attacks \cite{feng2019learning, huang2020unlearnable, sandoval2022poisons}. 
Availability attacks are a type of data poisoning, which allow the attacker to perturb the training dataset under a small norm restriction.
These attacks aim to cause test-time errors while maintaining the semantic integrity of the dataset and not affecting the normal usage of the data by legitimate users.
\citet{feng2019learning} used auto-encoder to generate adversarial poisons. \citet{huang2020unlearnable} proposed a bi-level error-minimizing approach to generate effective availability attacks for privacy preserving, which is called unlearnable examples. 
\citet{fowl2021adversarial} used stronger adversarial poisons to achieve availability attacks under error-maximizing approach.
%
\citet{yuan2021neural} used neural tangent kernels to generate error-minimizing poison noises.
\citet{yu2022availability} viewed availability attacks as a type of shortcut learning, and proposed the model-free synthetic noises to achieve comparable attack performance.
\citet{sandoval2022autoregressive} also proposed a model-free availability attack based on autoregressive perturbations.
\citet{wu2022one} expanded availability attacks to $\ell_0$ norm, generating poisons on one pixel, known as one-pixel shortcuts.
\citet{sadasivan2023cuda} generated effective unlearnable datasets based on class-wise convolution filters.
\citet{chen2023self} proposed self-ensemble protection to improve adversarial poisons by using several checkpoints.
In recent years, availability attacks have been generalized to adversarial training \cite{Fu2022RobustUE, wen2023is}, robustness \cite{tao2022can}, self-supervised learning \cite{he2022indiscriminate, ren2022transferable, wang2024efficient}, unsupervised learning \cite{zhang2023unlearnable}, natural language processing \cite{li2023make}.
There are also many work to defend availability attacks recently.
\citet{tao2021better} proposed an upper bound of clean test risk under defense of adversarial training.
\citet{qin2023learning, liu2023image} used stronger data augmentations to defend availability attacks more efficiently.
\citet{sandoval2023can} proposed an orthogonal projection method to defend availability attacks.
\citet{jiang2023unlearnable, dolatabadi2023devil} defended availability attacks by using a pre-trained diffusion model to purify injected poison noises.

\paragraph{Attacks on 3D point clouds.}
Vulnerability of 3D point cloud classifiers has become a grave concern due to the popularity of 3D sensors in applications.
Many works~\cite{xiang2019generating,cao2019adversarial,miao2022isometric} apply adversarial attacks to the 3D point cloud domain.
\citet{xiang2019generating} proposed point generation attacks by adding a limited number of synthetic points to the point cloud. 
Further research~\cite{huang2022shape,liu2022imperceptible,tsai2020robust,zheng2019pointcloud,zhang2022pointcutmix,miao2022isometric,zhang20233d,he2023generating,lou2024hide} has employed gradient-based point perturbation attacks.
However, adversarial attacks against 3D point cloud classifiers are all test-time evasion attacks. In contrast, our objective is to attack against the model at train-time.
The backdoor attack is another type of attack that poisons training data with a stealthy trigger pattern~\cite{chen2017targeted}.
\cite{li2021pointba} generated triggers on 3D point clouds by conducting orientation and interaction implanting for both poison-label and clean-label approach.
\cite{xiang2021backdoor} injected backdoors for 3D point clouds by optimizing local geometry and spatial location.
\cite{gao2023imperceptible} proposed imperceptible and robust backdoor attacks utilizing weighted local transformation.
However, the backdoor attack does not harm the model’s performance on clean data.

\section{Experiment Details}
\label{exp-setting}
\subsection{Datasets}

\paragraph{ModelNet40.}
ModelNet40 \cite{wu20153d} is a collection of CAD
models with 40 object categories, which is divided into 9843 training data and 2468 test data\footnote{https://shapenet.cs.stanford.edu/media/modelnet40\_ply\_hdf5\_2048.zip}.
\paragraph{ScanObjectNN.}
We use a real-world point cloud object dataset namely ScanObjectNN \cite{uy2019revisiting}, which contains 15 classes, and is divided into 2309 training data and 581 test data\footnote{https://hkust-vgd.ust.hk/scanobjectnn/}.
\paragraph{IntrA.} 
We use an open-access 3D intracranial aneurysm dataset IntrA \cite{yang2020intra}, which is a binary classification dataset divided into 1619 training data and 406 test data\footnote{https://drive.google.com/drive/folders/1yjLdofRRqyklgwFOC0K4r7ee1LPKstPh/IntrA.zip}.
\paragraph{BSM2017.} We use Basel Face Model 2017 \cite{gerig2018morphable} as our source face models to generate 3D point clouds of face scans\footnote{https://faces.dmi.unibas.ch/bfm/bfm2017.html}. We follow the setting in \cite{zhang2022learning} to generate 100 classes of face scans, each of them contains 50 point clouds. Furthermore, we divide the whole 5000 face data into training and test part, containing 4000 and 1000 data respectively.
After that, we randomly choose 1024 points for each point clouds before our training and poisoning approach.

\subsection{Networks}

\paragraph{PointNet.} 
PointNet \cite{qi2017pointnet} is a popular 3D point cloud classification network, and we use the architecture before the global feature as our feature extractor, the dimension of feature space is 1024.
\paragraph{PointNet++.}
We use the multi-scale grouping (MSG) network, a variant of PointNet, PointNet++ \cite{qi2017pointnet++}, which feature dimension is 1024.
\paragraph{DGCNN.}  DGCNN \cite{wang2019dynamic} uses EdgeConv blocks to obtain effective network for classification and segmentation of point cloud. The feature dimension is 2048 as we use both max pooling and average pooling when obtaining features.

\subsection{Details Setting of Algorithm}

\begin{table}[htb]
\centering
\caption{Detailed generation process of baselines and our FC-EM attack.}
\label{poison-params-details}
\begin{tabular}{llllllllllll}
\toprule[1pt] Poison  & FC-EM     & REG-EM  & REG-AP(-T) & EM   & AP(-T)    \\ 
\midrule  
Poison Batch Size       & 128    & 32    & 32   & 32    & 32     \\ 
Epochs & 200  & 200  & 100  & 200  & 100 \\
Learning Rate       & $10^{-3}$    & $10^{-3}$   & $10^{-3}$    & $10^{-3}$    & $10^{-3}$     \\ 
LR Schedule & Plateau & Plateau & Plateau & Plateau & Plateau \\
Dist Regularization & Chamfer & Chamfer & Chamfer & - & - \\
Reg Strength & 1.0 & 1.0 & 1.0 & - & - \\
$\ell_{\infty}$ Restraint & - & - & - & 0.08 & 0.08 \\
Temperature & 0.1 & - & - & - & - \\
PGD Steps & 10 & 10 & 250 & 10 & 250 \\
PGD Step Size & 0.015 & 0.015 & 0.001 & 0.015 & 0.001 \\

\bottomrule[1pt]
\end{tabular}
\end{table}

\paragraph{Adaptive regularization stregnth $\beta$.}
As shown in Table \ref{diff-beta-tab}, choosing the suitable $\beta$ is a tricky problem and could significantly influence the poison power and imperceptibility. 
Therefore, to mitigate this question, we may change $\beta$ adaptively at each epoch for different inputs.
In detail, we set an initial $\beta=1.0$, then we multiply/divide a scale factor $s$ at each epoch based on the $\L_{\fc}$ of each input. If $\L_{\fc}$ is large, we think the poison needs to focus more on similarity loss, otherwise, we think the poison should consider more on distance loss.
We set $s=1.1$, the threshold of $\L_{\fc}$ be 20\% largest inputs of the similarity loss, if $\L_{\fc}$ is top 20\%, we multiply $\beta$ with $s$, otherwise if $\L_{\fc}$ is last 80\%, we divide $\beta$ with $s$.
The maximum scale factor is 10, i.e, $\beta$ will change dynamicly in $[0.1, 10]$.

\section{More Experiments}
\label{app-exp}

\subsection{More Discussions on ScanObjectNN and OmniObject3D}
\label{scan-app}

Table \ref{scan-perf} showcases the results of experiments conducted on ScanObjectNN, revealing similar trends to those observed in ModelNet40. Our FC-EM attack outperforms $\ell_{\infty}$ based availability attacks, yielding the lowest test accuracy and better Chamfer and Hausdorff distances, which implies better imperceptibility of the poisoned point clouds.

The second row of Figure \ref{fig:fig5_1_2} provides visualizations of clean and various poisoned point clouds on ScanObjectNN. Notably, FC-EM poisoned point clouds maintain the structure with fewer outliers, while EM, AP, AP-T, and REG-AP-T exhibit more outliers and significant deformations and spurs, which are easily perceived by humans visually. Although REG-EM and REG-AP show decent imperceptiblity on point clouds, their poisoning efficacy, as indicated in Table \ref{scan-perf}, is compromised by relatively high test accuracy.

It is worth noting that some $\ell_{\infty}$ methods like AP-T and REG-AP-T exhibit quite strong poisoning abilities on ScanObjectNN, unlike their weaker counterparts on ModelNet40. This discrepancy may stem from ScanObjectNN being a scanned dataset rather than a generated one, making it inherently more challenging to learn, evident in its lower clean accuracy. Therefore, ScanObjectNN proves to be more susceptible to poisoning, allowing even relatively weaker poison methods to achieve commendable attack performance.

We also conduct additional experiments on the point cloud version of OmniObject3D dataset \cite{wu2023omniobject3d}, which has point clouds comprising 1024 points\footnote{https://openxlab.org.cn/datasets/OpenXDLab/OmniObject3D-New/tree/main/raw/point\_clouds/hdf5\_/1024}. We compare FC-EM with other availability attacks on this dataset using the PointNet architecture. as shown in Table \ref{app:OmniObject3D}, FC-EM yields lower test accuracy than the baseline methods under comparable Chamfer/Hausdorff distances, further validating the effectiveness of our proposed method.

\begin{table}[t]
\vspace{-1ex}
\caption{
Quantitative results of baselines and our FC-EM on the OmniObject3D dataset.
Our FC-EM attains the lowest test accuracy, inducing strongest effectiveness.
}
\label{app:OmniObject3D}
\small
\setlength{\tabcolsep}{5pt}
\centering
\begin{tabular}{lccccccccccc}
 \toprule
  Poison  & Acc & $\D_c$($\times 10^{-7}$) & $\D_h$($\times 10^{-6}$) \\
\midrule
  Clean  & 72.6\% & - & - \\
  EM  & 45.4\% & 39.2 & 9.2 \\
  AP  & 40.2\% & 35.8 & 9.1 \\
  AP-T & 45.6\% & 33.1 & 9.1 \\
  REG-EM & 66.4\% & 2.9 & 1.5  \\
  REG-AP & 50.2\% & 9.6 & 3.3 \\
  REG-AP-T &  53.5\% & 9.5 & 3.2 \\
  FC-EM (\textbf{ours}) & \textbf{20.1\%} & 11.3 & 6.9 \\
\bottomrule
\end{tabular}
\vspace{-1ex}
\end{table}

\subsection{Ablation Studies}
\textbf{Batch size.} 
 In the realm of 3D classification tasks, a common practice is to utilize a smaller batch size, like 32, during training \cite{qi2017pointnet}. 
 However, given that FC-EM incorporates class-wise similarity for poison generation, a larger batch size may be more suitable to ensure that each batch contains both positive and negative samples.
 Like in 2D image tasks, people typically use 128 batch size for supervised training but a larger one for contrastive training \cite{chen2020simple}.
 Motivated by this we set the batch size to be four times larger than the standard training, defaulting to 128.
 Nonetheless, we still assess poison performance across different batch sizes, as shown in Table \ref{metric-perf}. 
 The results indicate that FC-EM exhibits a relatively low sensitivity to variations in batch size, although 128 appears to be the most optimal choice of batch size.
\begin{table}[H]
\caption{Quantitative results on the test accuracy (Acc) of our FC-EM on different batch size. It demonstrates that FC-EM exhibits a relatively low sensitivity to variations in batch size.
}
\label{metric-perf}
\small
\centering
\begin{tabular}{lccccc}
 \toprule
Batch size  & 32 & 64 & 128 & 256 & 512\\
\midrule
  FC-EM & 29.01\% & 28.57\% & 27.67\% & 31.73\% & 36.91\% \\
\bottomrule
\end{tabular}
\end{table}
\vspace{-10pt}
\textbf{Different regularization strength.} 
To investigate the concrete impact of different distance regularization strength $\beta$, 
 we systematically evaluate FC-EM across various $\beta$ for comprehensive analysis, where the results are provided in Table \ref{diff-beta-tab}.
It is evident that with a smaller $\beta$, the distance regularization weakens, the poison performance enhances (lower test accuracy).  However, this improvement comes at the cost of larger Chamfer and Hausdorff distance, indicating poorer imperceptibility.
An pertinent question is the selection of an appropriate $\beta$ maintaining decent imperceptibility and achieving lower test accuracy. In practice, we found experimentally that on ModelNet40 when the Chamfer distance is below $10^{-3}$, the poisons appear nearly imperceptible. Conversely, a Chamfer distance exceeding $10^{-3}$ leads to severe deformations of the point clouds.

\begin{table}[!ht]
\caption{Quantitative results on the test accuracy (Acc) and distances (Chamfer distance ($\D_c$) and Hausdorff distance ($\D_h$)) of our FC-EM under different $\beta$.
With $\beta$ increases, the test accuracy gets higher, while the Chamfer distance and the Hausdorff distance get lower.}
\label{diff-beta-tab}
\small
\centering
\setlength{\tabcolsep}{4pt}
\begin{tabular}{lccccccccccccccc}
 \toprule
  $\beta$  & 0.1 & 0.2 & 0.4 & 0.6 & 0.8 & 1.0 & 1.5 & 2.0 & 3.0 & 5.0 & 10.0 & 20.0 \\
\midrule
 Test Acc & 19.73\% & 22.12\% & 25.36\% & 28.65\% & 27.19\% & 27.67\% & 28.32\% & 32.21\% & 41.45\% & 46.43\% & 56.52\% & 59.28\% \\
  $\D_c$($\times 10^{-4}$) & 67.3 & 37.8 & 19.7 & 12.7 & 10.2 & 8.4 & 7.4 & 6.7 & 4.9 & 4.3 & 2.7 & 2.0 \\
  $\D_h$($\times 10^{-3}$) & 119.9 & 75.7 & 41.9 & 21.6 & 19.7 & 13.5 & 12.4 & 11.7 & 8.1 & 7.6 & 5.3 & 4.5 \\
\bottomrule
\end{tabular}
\end{table}

\textbf{FC-EM under $\ell_{\infty}$ restriction.} 
Although FC-EM originates from breaking the degeneracy of the regularization process, the feature collision method can directly be applied under the classical $\ell_{\infty}$ norm restriction by simply modifying the loss to $\L_{\fc}$ during the poison generation. Therefore, we also evaluate the performance of FC-EM under $\ell_{\infty}$ restriction with budget $0.08$ in Table \ref{v-fc-em}. 

Results demonstrate that, even under $\ell_{\infty}$ restriction, FC-EM outperforms other methods, with the lowest test accuracy.
This superiority may stem from FC-EM inducing stronger linear separability, as proved in Theorem \ref{th-linear}, suggesting enhanced poison power.
However, it comprises with larger Chamfer and Hausdorff distance, potentially leading to a loss of imperceptibility. 
Furthermore, when FC-EM is applied under Chamfer restriction, it consistently outperforms its $\ell_{\infty}$ norm counterpart, highlighting the superiority of poisoning under distance regularization for 3D point clouds availability attacks.

\begin{table}[H]
\caption{Quantitative results on the test accuracy (Acc) and distances (Chamfer distance ($\D_c$) and Hausdorff distance ($\D_h$)) of our FC-EM attack under Chamfer and $\ell_{\infty}$ restriction.}
\label{v-fc-em}
\small
\centering
\begin{tabular}{lccccc}
 \toprule
  Restriction  & Acc & $\D_c$($
  \times 10^{-4}$) & $\D_h$($\times 10^{-3}$) &\\
\midrule
  Chamfer & 27.67\% & 8.4 & 13.5\\
  $\ell_{\infty}$ & 59.52\% & 42.0 & 18.7 \\
\bottomrule
\end{tabular}
\end{table}

\subsection{Poison Ratios}
We assess FC-EM when only a portion of the dataset is available for poisoning.
Results presented in Table \ref{diff-ratio} reveal that availability attacks degenerate significantly, even when only 1\% of data remains clean.
This limitation seems inherent to availability attacks, akin to the observed challenges in 2D image methods that effectiveness decreases when clean data is mixed \cite{huang2020unlearnable, fowl2021adversarial, sandoval2022autoregressive}. 

Availability attacks are typically formulated with the assumption that all training data can be poisoned, serving as a privacy-preserving mechanism. It is reasonable to expect that if clean (private) data is leaked, the efficacy of attacks will diminish.
However, developing a novel attack resilient to a small part of privacy leakage remains an intriguing problem, which may be left as the future work.

\vspace{-10pt}
 \begin{table}[H]
\caption{Quantitative results on the test accuracy (Acc) of our FC-EM under different poison ratios.
}
\label{diff-ratio}
\centering
\small
\setlength{\tabcolsep}{3pt}
\begin{tabular}{lccccccc}
 \toprule
  Ratio  & $99\%$ & $95\%$ & $90\%$ & $80\%$ & $60\%$ & $30\%$ \\
\midrule
  Acc(\%) & 65.24 & 75.73 & 79.42 & 82.29 &  85.62 & 87.07 \\
\bottomrule
\end{tabular}
\end{table}
\vspace{-10pt}

\subsection{Early Stopping}
We provide the learning process of the victim model when trained on FC-EM poisoned dataset in Figure \ref{acc_curve}.
It illustrates the rapid learning capability, with both training and validation accuracy approaching 100\% in the initial epochs.
Notably, in \cite{sandoval2022poisons}, it was observed that availability attacks often lose some power when early-stopping is applied. 
However, FC-EM exhibits strong resistance to early-stopping, signifying that it effectively deceive models from learning useful features, rather than initially acquiring and later being forgotten.
\vspace{-10pt}
 \begin{figure}[H]
 \centering
\includegraphics[width=8.0truecm]{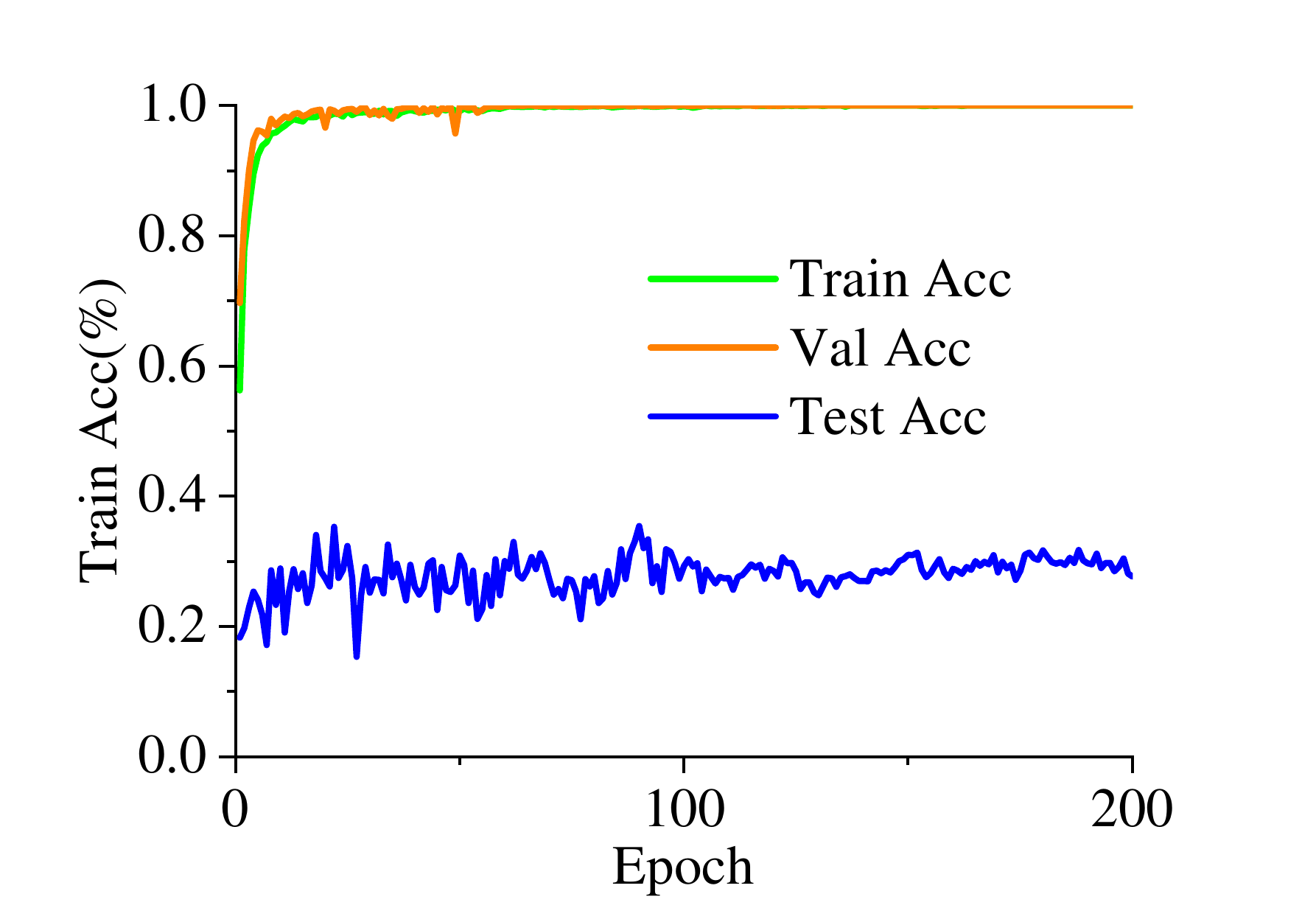}
\caption{Training, Validation and Test accuracy(\%) of FC-EM poisoned dataset on ModelNet40.
It illustrates that both training and validation accuracy approaching 100\% in the initial epochs, but the test accuracy remains low constantly.}
 \label{acc_curve}
\end{figure}



\subsection{Experiments on More Network Structures}

There are more advanced 3D recognition network architectures\cite{xu2021learning,xu2021paconv,zhao2021point,ma2022rethinking}.
To further validate the effectiveness of FC-EM, we evaluate it alongside other availability attacks on two modern point cloud classification networks, PCT \cite{zhao2021point} and PointMLP \cite{ma2022rethinking}. Results are shown in Table \ref{modelnet40-perf-more}. The results reveal that the performance of availability attacks on these networks aligns with those observed on PointNet, PointNet++, and DGCNN, with our FC-EM showing consistent efficacy and decent distances. This suggests that FC-EM's performance is largely independent of the underlying model architecture, making it applicable to both traditional and modern networks. 

\begin{table}[H]
\caption{
Quantitative results of baselines and our FC-EM on ModelNet40 on PCT and PointMLP. 
Our FC-EM achieves the lowest test accuracy (Acc) and exhibits suitable imperceptibility (measured by Chamfer distance ($\D_c$) and Hausdorff distance ($\D_h$)), causing the strongest availability attack.
}
\small
\label{modelnet40-perf-more}
\centering
\setlength{\tabcolsep}{5pt}
\begin{tabular}{lccccccccc}
 \toprule
  & \multicolumn{3}{c}{PCT} & \multicolumn{3}{c}{PointMLP} \\
  Poison  & Acc & $\D_c$($\times 10^{-4}$) & $\D_h$($\times 10^{-3}$) & Acc & $\D_c$($\times 10^{-4}$) & $\D_h$($\times 10^{-3}$)  \\
\midrule
  Clean & 92.02\% & - & - & 93.84\% & - & - \\
    EM & 80.31\% & 24.2 & 11.2 & 72.20\% & 26.7 & 13.9 \\
  AP & 67.18\% & 17.5 & 13.5 & 61.43\% & 18.5 & 11.6 \\
  AP-T & 67.02\% & 18.4 & 14.0 & 45.34\% & 16.7 & 12.1 \\
  REG-EM & 80.96\% & 4.5 & 5.6 & 85.86\% & 3.4 & 3.0 \\
  REG-AP & 72.89\% & 6.2 & 12.7 & 87.56\% & 7.3 & 10.5 \\
  REG-AP-T & 71.64\% & 7.8 & 13.3 & 69.41\% & 2.7 & 7.6 \\
  FC-EM (\textbf{ours}) & \textbf{34.32\%} & 5.5 & 7.5 & \textbf{21.88\%} & 6.4 & 6.7 \\
\bottomrule
\end{tabular}
\vspace{-2ex}
\end{table}

\section{Degeneracy of Adversarial Poisons}
\label{ap-degrad}
Error-maximization poisoning approach, known as Adversarial Poisons \cite{fowl2021adversarial}, was designed by training a clean source model and crafting poisons using 250 steps of PGD on the loss-manimization objective.
It is worth noting that, for common adversarial training/attacks, people often use a much smaller steps of PGD, like PGD-10 and PGD-20, to obtain decent defense/attack performance.
Differently, AP conducts a larger step of PGD attack, in other word, it use stronger adversarial attacks to achieve availability attacks. 

Unfortunately, when we add distance regularization on such stronger attack approach, i.e., we directly extend AP to REG-AP, the distance regularization term will force the poison $\delta$ to be a weaker adversarial attack rather than a stronger one.
Intuitively, distance regularization will lend $\delta$ become smaller while maintaining the effective attack. Therefore, REG-AP tends to find a successful adversarial attack under a smaller distance, rather than find a stronger adversarial attack within a attack restriction region.

Precisely, the following theorem demonstrates that, when using AP attacks, if data regularization is conducted, the convergent point is either not an adversarial point, or a weakest adversarial point that possesses the smallest distance.
Therefore, choose larger attack step to obtain stronger adversarial poisons will fail if distance regularization exists.
\begin{theorem}
Let $\L_{\cls}$ be the 0-1 loss. 
The optimal point $\delta^*$ of the untargeted optimization \eqref{err-max-equ} satisfies either (1) $\delta^*=\argmin\limits_{\delta\in\Delta}\L_{\dist}(x+\delta, x)$, where $\Delta=\{\delta\sep f_{\theta}(x+\delta)\neq y\}$, or (2)
$f_{\theta}(x+\delta^*)=y$. 
Similarly, the optimal point $\delta^*$ of the targeted optimization \eqref{err-max-equ} satisfies (1) $\delta^*=\argmin\limits_{\delta\in\Delta}\L_{\dist}(x+\delta, x)$, where $\Delta=\{\delta\sep f_{\theta}(x+\delta)=\tau(y)\}$, or (2)
$f_{\theta}(x+\delta^*)\neq\tau(y)$.
\end{theorem}
\begin{proof}
We only prove the untargeted case, the targeted one can be proved similarly. If the optimal point $\delta^*$ satisfies that $f_{\theta}(x+\delta)\neq y$, as $\L_{\cls}$ is the 0-1 loss, $\L_{\cls}(x+\delta, y, \theta^*)=1$ for all $\delta\in\Delta$.

Therefore, $\delta^*=\argmin\limits_{\Delta}\L_{\dist}(x+\delta, x)$.
For any successful adversarial attacks $\delta$,  optimization \eqref{err-max-equ} will force to find the smallest magnitude of perturbation $\delta$ such that $\L_{\dist}$ achieves minimum.
\end{proof}

\section{Feature Disentanglement of Adversarial Poisons}
\label{fd-ap}



We also add similarity loss for generation of AP poisons. It is worth noting that when generating poisons by maximizing feature collision loss, one could cause inner-class features disperse while inter-class features close. Therefore, the feature used for classification will be disentangle under maximization, we called it \textit{Feature Disentanglement Error-Maximization} (FD-AP) poison. Moreover, AP process generates poison noises that exhibits non-robust features in incorrect class, but $\L_{\fc}$ itself does not include label information, we maximize both $\L_{\cls}$ and $\L_{\fc}$ in practice for effective poisons.

We formulate the FD-AP process as below:
\begin{align} \label{fc-ap-equ}
    &\max_{{\delta}} \E_{(x, y) \in D }\big[ \L_{\cls}(x+\delta,y;\theta^*) + \zeta \cdot \L_{\fc}(x+\delta,y;\theta^*, t) - \beta \cdot \L_{\dist} (x+\delta, x)\big],  \hbox{ (Untargeted)} \nonumber\\
    \textup{or} \quad &\min_{{\delta}} \E_{(x, y) \in D }\big[ \L_{\cls}(x+\delta,\tau(y);\theta^*) + \zeta \cdot \L_{\fc}(x+\delta,\tau(y);\theta^*, t) + \beta \cdot \L_{\dist} (x+\delta, x)\big],  \hbox{ (Targeted)} \nonumber\\
    & s.t.\ \ \ \  \theta^* \in \argmin_{\theta} \E_{(x, y) \in D }\big[ \L_{\cls}(x, y;\theta) 
    ],
\end{align}
where $\zeta$ is the hyperparameter controlling the strength of feature collision loss, the detailed algorithm is provided in Algorithm \ref{alg:fd-ap}.

Experimental results in Table \ref{knn-reg} show that our FD-AP method slightly outperform all of the other AP-based methods, without losing too much distance compared to REG-AP(-T), and strictly be better than AP(-T).
However, FD-AP performs suboptimally compared with FC-EM. The main reason may result from that feature similarity operation needs more epoch to optimize, like in contrastive learning \cite{chen2020simple}, they often require 1000 epochs rather than 100/200 epoch for standard training. 
But FD-AP instead, only generate poisons once with more PGD steps, which may undermine the utility of feature similarity. 
A better approach to convert AP-based methods to 3D point clouds is an interesting question. We may leave it as the future work.

\begin{algorithm}[htb]
\caption{Feature Disentanglement 3D Adversarial Poison Attack (FD-AP)}
\label{alg:fd-ap}
\begin{algorithmic}
\STATE {\bfseries Input:} A 3D point cloud training dataset $D=\{(x_i, y_i)\}_{i=1}^N$. A model with initialized parameters $\theta$.  Distance loss function $\L_{\dist}$ and regularization strength $\beta$.
Feature collision loss function $\L_{\fc}$, strength $\zeta$ and temperature $t$. 
Classifier parameters $\alpha_{\theta}$ and $T_\theta$. Attack parameters $\alpha_a$ and $T_a$.   
Label group transformation $\tau$ for FD-AP-Targeted.

\STATE {\bfseries Output:} 
Poisoned dataset $D_{\delta}=\{ (x_i + \delta_i, y_i) \}_{i=1}^N$

\STATE $\delta_i\gets 0, i=1,2,\cdots,N$ 

\FOR{$t_\theta = 1,\cdots, T_\theta$} 
\STATE Sample a mini batch $B=\{x_{b_j}, y_{b_j}\}_{j=1}^{N_B}$
\STATE $\theta \gets \theta - \alpha_{\theta} \cdot \nabla_\theta \E_{(x_{b_j}, y_{b_j}) \in B } \big[ \L_{\cls}(x_{b_j}+\delta_{b_j},y_{b_j};\theta)\big]$
\ENDFOR

\FOR{$i=1,2,\cdots,N$}
    \FOR{$t_a = 1,\cdots, T_a$}
    \STATE $\delta_N \gets \delta_N + \alpha_a \cdot \nabla_{\delta_N} \E_{(x_N, y_N) \in D } 
    \big[ \L_{\cls}(x_N+\delta_N,y_N;\theta) + \zeta \cdot \L_{\fc}(x_N+\delta_N,y_N;\theta, t) - \beta \cdot \L_{\dist} (x_N+\delta_N, x_N)\big]$ (\textup{Untargeted}) \\
    \STATE $\delta_N \gets \delta_N - \alpha_a \cdot \nabla_{\delta_N} \E_{(x_N, y_N) \in D } 
    \big[\L_{\cls}(x_N+\delta_N,\tau(y_N);\theta) + \zeta \cdot \L_{\fc}(x_N+\delta_N,\tau(y_N);\theta, t) + \beta \cdot \L_{\dist} (x_N+\delta_N, x_N)\big]$ (\textup{Targeted}) \\
    \ENDFOR
\ENDFOR

\end{algorithmic}
\end{algorithm}

 \begin{table}[H]
\caption{Quantitative results on the test accuracy (Acc) on PointNet under various Error-maximum (AP) based poison methods for ModelNet40.
Our proposed FD-AP-T obtains the lowest test accuracy (Acc) and gains good imperceptibility (measured by Chamfer distance ($\D_c$) and Hausdorff distance ($\D_h$)), outperforming all other AP-based availability attacks.
}
\label{knn-reg}
\centering
\small
\begin{tabular}{lcccc}
 \toprule
  Poison  & Acc & $\D_c$($\times 10^{-4}$) & $\D_h$($\times 10^{-3}$)\\
\midrule
  AP   &  69.37\% & 27.3 & 17.0  \\
  AP-T   &  61.95\% & 26.7 & 16.8  \\
  REG-AP   & 75.32\% & 5.5 &  11.5 \\
  REG-AP-T  & 72.61\% & 5.1 & 11.5 \\
    FD-AP  & 61.71\% & 11.8 & 17.1 \\
    FD-AP-T & 52.88\% & 7.0 & 13.4 \\
\bottomrule
\end{tabular}
\end{table}

\end{document}